\documentclass[11pt, a4paper]{article} 

\title{Semi-Algebraic Proof Systems for QBF
\thanks{An extended abstract of this article has appeared in proceedings of SAT'25 \cite{BBKMS-SAT-25}. Olaf Beyersdorff and Kaspar Kasche were supported by Carl-Zeiss Foundation. Olaf Beyersdorff was additionally supported by DFG grant BE 4209/3-1. Ilario Bonacina was funded by the AEI with the grant number PID2022-138506NB-C22. Meena Mahajan was partially supported by the J. C. Bose Fellowship of SERB, ANRF.}
}

\author{
Olaf Beyersdorff
\\
\small Friedrich Schiller University 
\\ \small Jena, Germany 
\\ \small \texttt{olaf.beyersdorff@uni-jena.de}
\and
Ilario Bonacina
\\
\small Universitat Politecnica de Catalunya
\\ \small Barcelona, Spain 
\\ \small \texttt{ilario.bonacina@upc.edu}
\and
Kaspar Kasche
\\ \small Friedrich Schiller University 
\\ \small Jena, Germany 
\\ \small \texttt{kaspar.kasche@uni-jena.de}
\and
Meena Mahajan
\\ \small The Institute of Mathematical Sciences 
\\ \small (A CI of Homi Bhabha National Institute)
\\ \small Chennai, India
\\ \small \texttt{meena@imsc.res.in}
\and
Luc Nicolas Spachmann
\\ \small Friedrich Schiller University
\\ \small Jena, Germany 
\\ \small \texttt{luc.spachmann@uni-jena.de}
}

\usepackage[noadjust]{cite}
\usepackage{subcaption}
\usepackage{adjustbox}
\usepackage{cite}
\usepackage{hyperref}
\hypersetup{
colorlinks= true,
linkcolor=blue!60!black,
citecolor=blue!60!black,
pagebackref=true,
}
\usepackage[left=3cm, right=3cm]{geometry}

\usepackage{amsthm}
\usepackage{amsmath,amssymb}
\usepackage{thmtools}
\newtheorem{thm}{Theorem}[section]
\newtheorem{lem}[thm]{Lemma}
\newtheorem{cor}[thm]{Corollary}
\newtheorem{prop}[thm]{Proposition}
\theoremstyle{definition}
\newtheorem{defi}[thm]{Definition}
\theoremstyle{remark}
\newtheorem{fact}[thm]{Fact}
\newtheorem{rem}[thm]{Remark}

\usepackage{cleveref}
\crefname{thm}{Theorem}{Theorems}
\Crefname{thm}{Theorem}{Theorems}
\crefname{defi}{Definition}{Definitions}
\Crefname{defi}{Definition}{Definitions}
\crefname{rem}{Remark}{Remarks}
\Crefname{rem}{Remark}{Remarks}
\crefname{cor}{Corollary}{Corollaries}
\Crefname{cor}{Corollary}{Corollaries}
\crefname{lem}{Lemma}{Lemmas}
\Crefname{lem}{Lemma}{Lemmas}
\crefname{fact}{Fact}{Fact}
\Crefname{fact}{Fact}{Fact}
\crefname{prop}{Proposition}{Proposition}
\Crefname{prop}{Proposition}{Proposition}

\usepackage{xspace}
\usepackage{tikz}
\usetikzlibrary{arrows.meta}
\usetikzlibrary{shapes.geometric}
\makeatletter
\newbox\dottedarrow@box
\setbox\dottedarrow@box\hbox
  {%
    \begin{tikzpicture}
      \draw[dotted,->] (0,0) -- (1.5em,0);
    \end{tikzpicture}%
  }
\newcommand*\dottedarrow
  {\relax\ifmmode\expandafter\dottedarrow@m\else\expandafter\dottedarrow@t\fi}
\newcommand*\dottedarrow@t[1][1.5em]
  {\resizebox{#1}{!}{\raisebox{.5ex}{\usebox\dottedarrow@box}}}
\newcommand*\dottedarrow@m[1][]
  {%
    \if\relax\detokenize{#1}\relax
      \mathchoice
        {\dottedarrow@t}
        {\dottedarrow@t}
        {\dottedarrow@t[1.1em]}
        {\dottedarrow@t[0.9em]}%
    \else
      \dottedarrow@t[#1]%
    \fi
  }
\makeatother

\renewcommand{\vec}{\boldsymbol}

\DeclareMathOperator{\enc}{enc}
\renewcommand{\phi}{\varphi}
\renewcommand{\le}{\leqslant}
\renewcommand{\ge}{\geqslant}
\renewcommand{\leq}{\leqslant}
\renewcommand{\geq}{\geqslant}
\newcommand{\ns}[1][]{\mathsf{NS}_{#1}\xspace}

\newcommand{\SA}[1][]{\mathsf{SA}_{#1}\xspace}
\newcommand{\sos}[1][]{\mathsf{SOS}_{#1}\xspace}
\newcommand{\res}{\textsf{Res}\xspace}
\newcommand{\pc}[1][]{\mathsf{PC}_{#1}\xspace}
\newcommand{\fregeTC}{\mathsf{TC}_0\text{-}\mathsf{Frege}}
\newcommand{\qns}[1][]{\ensuremath{\mathsf{Q}\text{-}\ns[#1]}\xspace}

\newcommand{\qsa}{\ensuremath{\mathsf{Q}\text{-}\SA}\xspace}
\newcommand{\qsos}{\ensuremath{\mathsf{Q}\text{-}\sos}\xspace}

\newcommand{\qures}{\ensuremath{\mathsf{QU}\text{-}\res}\xspace}
\newcommand{\qpc}{\ensuremath{\mathsf{Q}\text{-}\pc}\xspace}

\newcommand{\qtcz}{\textsf{Q-TC$_0$-Frege}\xspace}

\newcommand{\wres}{\textsf{w-Res}\xspace}
\newcommand{\reswres}{\textsf{restricted-w-Res}\xspace}
\newcommand{\qwres}{\textsf{Q-w-Res}\xspace}

\let\E\undefined
\DeclareMathOperator{\E}{{\mathbb{E}}}
\newcommand{\psE}{\tilde{\mathbb{E}}}

\newcommand{\restr}[2]{\left.{#1}\right|_{#2}}

\newcommand{\equality}{\textsf{Equality}\xspace}
\newcommand{\parityformula}{\textsf{Parity}\xspace}
\newcommand{\maj}{\textsf{Majority}\xspace}
\newcommand{\qmaj}{\textsf{Q-Majority}\xspace}

\newcommand{\abs}[1]{\lvert #1 \rvert}

\DeclareMathOperator{\qdeg}{qdeg_\exists}

\DeclareMathOperator{\vars}{vars}
\DeclareMathOperator{\score}{score}
\DeclareMathOperator{\qsize}{qsize}
\DeclareMathOperator{\size}{size}
\DeclareMathOperator{\sign}{sign}

\newcommand{\Ind}{\textsf{Ind}}

\begin{document}

\maketitle

\begin{abstract}
  \noindent We introduce new semi-algebraic proof systems for Quantified Boolean Formulas (QBF) analogous to the propositional  systems Nullstellensatz, Sherali-Adams and Sum-of-Squares. We  transfer to this setting techniques both from the QBF literature (strategy extraction) and from propositional proof complexity  (size-degree relations and pseudo-expectation). We obtain a number of  strong QBF lower bounds and separations between these systems, even when disregarding propositional hardness.
\end{abstract}

\section{Introduction}
\label{sec:intro}

Two key results in algebraic and semi-algebraic geometry are the  Nullstellensatz and the Positivstellensatz. 
The first can be seen as an algebraic identity certifying that a set of polynomial equations is unsatisfiable while the second as an algebraic identity certifying that a system of polynomial \emph{inequalities} is unsatisfiable. 
In other words, both the Nullstellensatz and the Positivstellensatz naturally give rise  to proof systems and in recent years intense research was performed on the proof complexity of such systems. 
In particular, the proof system \emph{Sum-of-Squares}, a special case of  Positivstellensatz, starting from \cite{BBHKSZ12}, has received a lot of attention for its connection with algorithms based on hierarchies of SDP relaxations (see for instance \cite{Lasserre01,ODonnellZ13}). For similar reasons, the even more restrictive \emph{Sherali-Adams} proof system has been investigated, named after its connection with the Sherali-Adams hierarchy of linear programming since its original definition  (see for instance \cite{SheraliA90,DantchevM13}).

In this work, we devise a simple and natural way to extend the proof systems Nullstellensatz~($\ns$), Sum-of-Squares ($\sos$) and Sherali-Adams ($\SA$) from the context of propositional existentially quantified variables to existentially and universally quantified variables. In other words, we show how to define  proof systems for quantified Boolean formulas (QBF)  inspired by the propositional proof systems above.

 The study of propositional and QBF proof systems is motivated both by theoretical reasons and also by  connections to SAT and QBF solving \cite{sathandbookpc,qbfhandbook,M4CQBF}.  
 While for SAT-solvers conflict-driven clause learning (CDCL) is the ruling paradigm, which by the seminal work of \cite{AtseriasFT11,DBLP:journals/ai/PipatsrisawatD11} is essentially equivalent to the propositional proof system \emph{Resolution} ($\res$), there are several competing approaches in QBF, with CDCL-based \cite{ZM02} and \emph{expansion}-based solving \cite{JM15} among the main paradigms. 
 To model the strength of QBF-solvers several (often incomparable) QBF proof systems have been introduced and analysed \cite{ELW13,JM15,BB23-LMCS}. Of  most relevance to this work is QU-Resolution ($\qures$) \cite{KBKF95,Gelder12}. 
 $\qures$ adds to propositional Resolution the $\forall$-reduction rule  that allows to eliminate universal variables from clauses. In \cite{BBCP20} it has been shown that this approach of augmenting  a propositional proof system by  a $\forall$-reduction rule taking care of  universal quantifiers also works for other common inference-based  proof systems such as various Frege systems \cite{CR79,BBCP20}, \emph{Cutting Planes} \cite{CCT87,BCMS18-CP} ---a proof system modelling geometric reasoning related to Chv\'atal-Gomory cuts--- and \emph{Polynomial Calculus} \cite{CEI.96,BeyersdorffHKS24}, modelling algebraic reasoning related to Gr\"obner bases computations.
 
 A common feature of all the propositional proof systems above is that they are \emph{inference}-based, unlike the static systems $\ns$/$\SA$/$\sos$ where a proof is just an algebraic identity of some specific form depending on the system at hand. 
 This has posed quite some problems to adapt such systems to the QBF setting, and it was not clear at all whether an approach similar to a $\forall$-reduction rule was even viable. Recently there has been a  suggestion to define QBF analogues of $\ns$ based on  $\forall$-expansion \cite{CCKS23}, but these differ considerably from our approach here and the $\forall$-reduction paradigm discussed above.   
 
 We show that an approach similar to a $\forall$-reduction rule \emph{does} allow to define QBF versions of $\ns$/$\SA$/$\sos$, which we call $\qns$/$\qsa$/$\qsos$ respectively. We argue that our definitions are quite natural: they add to the algebraic equations of $\ns$/$\SA$/$\sos$ simple polynomials that strongly resemble $\forall$-reduction and meet the same technical condition on variable dependence.

We begin the systematic study of these QBF proof systems in terms of lower and upper bounds, strategy extraction, and simulations.  
 Concerning the latter, \cref{fig:p-sim-a} recalls the relations between propositional $\ns$/$\SA$/$\sos$ and further propositional proof systems such as Resolution ($\res$) and Polynomial Calculus ($\pc$). In  \cref{fig:p-sim-b}, we depict our results on the new QBF systems $\qns$/$\qsa$/$\qsos$ and how they relate to $\qures$ and $\qpc$. The figures are virtually identical: what changes is that proofs of the simulations, although mimicking those in the propositional setting, require  extra care. \cref{fig:p-sim-c} depicts the simulation order  when we factor out the propositional complexity and consider \emph{genuine} QBF hardness stemming from quantifier alternations -- a framework that has become standard in QBF proof complexity (cf.\ \cite{BHP20,Che17} for background). We call the `genuine' size measure qsize, which only counts monomial size in the new $\forall$-reduction polynomials. Lower bounds on qsize are tighter and trivially imply lower bounds on the traditional size measure that counts all monomials. Hence, lower bounds and separations in qsize are harder to obtain. In fact, $\qsa$ and $\qsos$ become equivalent w.r.t. qsize while they are separated w.r.t.\ size.
 
 The fact that the systems we define fit so nicely into the lattice of QBF proof systems using the $\forall$-reduction approach suggests that the definitions we give are natural analogues of  $\forall$-reduction in this context. The analogy with the $\forall$-reduction rule of $\qures$ gets even  clearer when using the language of weighted clauses and Resolution to describe $\ns$ and $\SA$ \cite{BBL.24}. For simplicity, we describe $\qns$ and $\qsa$ first using the usual algebraic language, and in \Cref{sec:p-simulations-extra} we describe $\qsa$ also in the language of weighted clauses. 

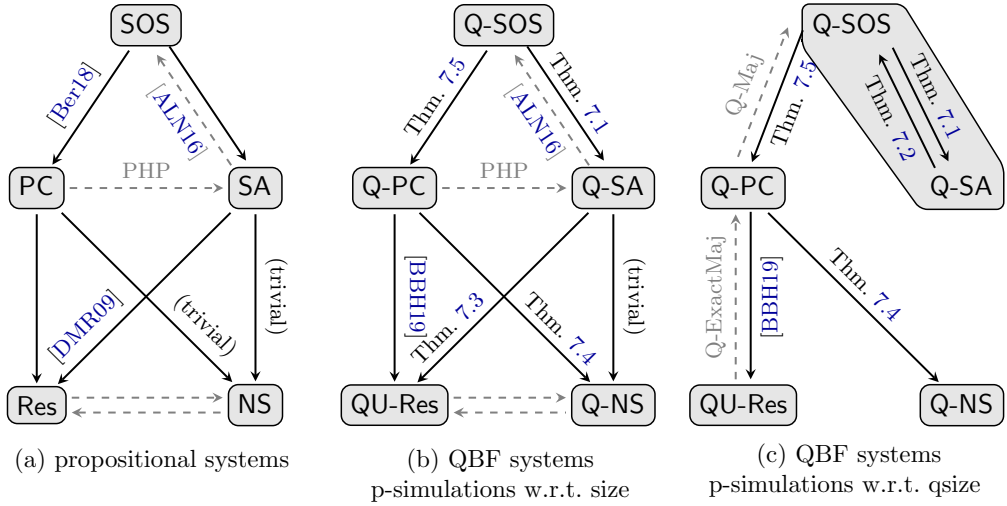
\begin{figure}[h!]
\captionsetup[subfigure]{justification=centering}
\usetikzlibrary{calc}
\usetikzlibrary{backgrounds}
\centering
\tikzset{
every node/.append style={semithick, rounded corners, align=center, node distance=8em},
psim/.style = {thick, ->, >=stealth, shorten >=1pt, shorten <=2pt},
not psim/.style ={psim, dashed, color=black!50!white, shorten >=2pt, shorten <=2pt},
}

\newcommand\SingleLine[5][auto, sloped, near end]{
    \draw[#4, ->](#2)-- node [#1] {\footnotesize #5} (#3); 
    }

\newcommand\DoubleLine[7][3pt]{%
    \path(#2)--(#3)coordinate[at start](h1)coordinate[at end](h2);
    \draw[#4, ->]($(h1)!#1!90:(h2)$) -- node[above,sloped] {\footnotesize #5} ($(h2)!#1!-90:(h1)$); 
    \draw[#6, <-]($(h1)!#1!-90:(h2)$)-- node [below,sloped] {\footnotesize #7} ($(h2)!#1!90:(h1)$);
    }
 \begin{adjustbox}{max width=.9\textwidth}
\begin{subfigure}[t]	{.3\textwidth}
	\centering
	\begin{tikzpicture}
		\begin{scope}[every node/.append style={draw,fill=black!10!white}]
		\node (ns) at (0,0){$\ns$};
		\node[left of=ns] (res) {$\res$};
		\node[above of=ns] (sa) {$\SA$};
		\node[left of =sa] (pc) {$\pc$};
		\node (sos) at ([yshift=6em]$(pc)!0.5!(sa)$) {$\sos$};
		\end{scope}
		\DoubleLine{res}{ns}{not psim}{}{not psim}{}
		\SingleLine[sloped, above]{sa}{ns}{psim}{(trivial)}
		\DoubleLine{sos}{sa}{psim}{}{not psim}{\cite{ALN.16}}
		\SingleLine{sa}{res}{psim}{\cite{DMR.09}}
		\SingleLine[sloped,above]{pc}{sa}{not psim}{PHP}
		\SingleLine{pc}{res}{psim}{}
		\SingleLine{pc}{ns}{psim}{(trivial)}
		\SingleLine[sloped,above]{sos}{pc}{psim}{\cite{Berkholz18}}
	\end{tikzpicture}
\caption{propositional systems\label{fig:p-sim-a}}
\end{subfigure}
\hspace{0.01\textwidth}
\begin{subfigure}[t]	{.3\textwidth}
\begin{tikzpicture}
		\begin{scope}[every node/.append style={draw,fill=black!10!white}]
		\node (ns) at (0,0){$\qns$};
		\node[left of=ns] (res) {$\qures$};
		\node[above of=ns] (sa) {$\qsa$};
		\node[left of =sa] (pc) {$\qpc$};
		\node (sos) at ([yshift=6em]$(pc)!0.5!(sa)$) {$\qsos$};
		\end{scope}
		\DoubleLine{res}{ns}{not psim}{}{not psim}{}
		\SingleLine[sloped,above]{sa}{ns}{psim}{(trivial)}
		\DoubleLine{sos}{sa}{psim}{Thm.~\ref{thm:psim-sos-sa}}{not psim}{\cite{ALN.16}}
		\SingleLine{sa}{res}{psim}{Thm.~\ref{thm:psim-sa-res}}
		\SingleLine[sloped,above]{pc}{sa}{not psim}{PHP}
		\SingleLine[sloped,above]{pc}{res}{psim}{\cite{BBH19}}
		\SingleLine{pc}{ns}{psim}{Thm.~\ref{thm:psim-pc-ns}}
		\SingleLine[sloped,above]{sos}{pc}{psim}{Thm.~\ref{thm:psim-sos-pc}}
\end{tikzpicture}
\caption{QBF systems\\p-simulations w.r.t.\ $\size$\label{fig:p-sim-b}}
\end{subfigure}
\hspace{0.01\textwidth}
\raisebox{2pt}{%
\begin{subfigure}[t]	{.3\textwidth}
\begin{tikzpicture}
		\node[draw,fill=black!10!white] (ns) at (0,0){$\qns$};
		\node[left of=ns,draw,fill=black!10!white] (res) {$\qures$};
		\node[above of=ns] (sa) {$\qsa$};
		\node[above of =res, draw,fill=black!10!white] (pc) {$\qpc$};
		\node (sos) at ([yshift=6em]$(pc)!0.5!(sa)$) {$\qsos$};
		\begin{scope}[on background layer]
		\draw[rounded corners=1mm,semithick,fill=black!10!white] (sos.north east) -| (sos.west) |- (sos.south west) -- (sa.south west) -- (sa.south) -| (sa.north east) --cycle;
		\end{scope}
		\DoubleLine{res}{pc}{not psim}{Q-ExactMaj}{psim}{\cite{BBH19}}
		\SingleLine[sloped,above]{pc}{ns}{psim}{Thm.~\ref{thm:psim-pc-ns}}
		\DoubleLine{[yshift=.1em]pc.north}{[xshift=-.2em]sos.west}{not psim}{Q-Maj}{psim}{Thm.~\ref{thm:psim-sos-pc}}
		\DoubleLine{[xshift=-.7em]sos.south east}{[xshift=-.7em]sa.north}{psim}{Thm.~\ref{thm:psim-sos-sa}}{psim}{Thm.~\ref{thm:pequiv-semialg}}
		\end{tikzpicture}
\caption{QBF systems\\p-simulations w.r.t.\ $\qsize$ \label{fig:p-sim-c}}
\end{subfigure}
}
\end{adjustbox}
\caption{
\label{fig:p-sim}
Simulations and separations between algebraic proof systems in the propositional and the QBF setting. 
By $P \rightarrow Q$
 we indicate that  proof system $P$ polynomially simulates $Q$, while $ P\ \textcolor{black!50!white}{\dashrightarrow}\ Q$
 means that the proof system $P$ does not polynomially simulate $Q$. We omit  polynomial (non-)simulations  implied by those displayed.
}
\end{figure}

For lower and upper bounds we develop and adapt  techniques that originate both from propositional and QBF proof complexity. 

Regarding the transfer of propositional techniques, we show how to lift common techniques for $\sos$ from the propositional setting to $\qsos$: we establish a size-degree relation analogous to the propositional one \cite{AH.19}, and  show how to adapt the notion of pseudo-expectation, the prime lower-bound method for semi-algebraic systems \cite{FKP.19} to the QBF setting. Both adaptations require interesting modifications and do not just replicate the propositional techniques (see \cref{lemma:sizedegree_helper} and \cref{def:Q-pseudoexpectation} with the discussion thereafter). We use pseudo-expectations to show an exponential lower bound for $\qsos$ for the $\equality$ QBFs \cite{BBH18} with respect to the tighter qsize measure (\cref{thm:Equality}).

Regarding QBF techniques, we develop strategy extraction for $\qsos$. Strategy extraction has become the predominant technique to analyse QBF proof systems (see~\cite{BeyersdorffCJ19,BBCP20,BBH19} for instance) and is also of tremendous practical importance for QBF solving and verification \cite{qbfhandbook,ELW13,Balabanov12}. Specifically, we show that $\qsos$ allows strategy extraction by polynomial threshold functions and develop a new score game interpretation. Interestingly, this  score game allows to characterise genuine proof size in $\qsos$ and $\qns$ (Theorem~\ref{thm:game_equivalence}). We also use it to elegantly show completeness of the new systems and a linear upper bound for the $\qmaj$ QBFs, which are known to be hard for $\qpc$ \cite{BeyersdorffHKS24},  yielding the separation depicted in \cref{fig:p-sim-c}. 

\paragraph{Structure of the paper.}
\Cref{sec:preliminaries} contains preliminaries and notation.
\Cref{sec:algQBF} defines our new semi-algebraic QBF systems and shows their  soundness and completeness together with the size-degree relation for $\qsos$.
\Cref{sec:strategy} contains the strategy extraction for $\qsos$ and some consequences.
In \Cref{sec:pseudoexpectation} we develop the lower bound technique of pseudo-expectations for $\qsos$ and show an exponential lower bound.
\Cref{sec:p-simulations-extra} shows a lifted weighted resolution to QBF and its equivalence to $\qsa$.
In \Cref{sec:p-simulations} we compare the QBF systems via p-simulations.
We conclude in \Cref{sec:conclusions} with some  open problems.

\section{Preliminaries}
\label{sec:preliminaries}

\paragraph{QBF preliminaries.}
We consider Quantified Boolean Formulas (QBF) of the form $\mathcal{Q}. \phi$, where $\phi$ is a CNF formula and $\mathcal Q$ is the quantifier prefix. Both the variables of $\phi$ and the variables of $\mathcal Q$ range over a set of Boolean variables $V$. Let $\mathrm{vars}_\forall(\mathcal Q)$ (resp.\ $\mathrm{vars}_\exists(\mathcal Q)$) be the set of universally (resp.\ existentially) quantified variables in $\mathcal Q$. 

The evaluation of a QBF formula $\mathcal Q.\phi$ can be seen as a game (the \emph{evaluation} game) between two players: the existential $\exists$-player and the universal $\forall$-player, where the $\exists$-player's goal is to satisfy the formula $\phi$ and the $\forall$-player's goal is to falsify it. The players take turns according to the order of the quantifiers in $\mathcal Q$. We call this game the evaluation game to distinguish it from a new game, the \emph{score} game, which we introduce in \Cref{sec:completeness}.

A very well-studied QBF proof system  is $\qures$ \cite{BWJ14,KBKF95,Gelder12}, which can be seen as a natural extension of the propositional proof system Resolution \cite{Bla37,Rob65} to the QBF setting. 

The QBF proof system $\qures$ refutes a false QBF $\mathcal Q.\phi$ inferring the empty clause $\bot$ from the clauses in $\phi$ using the resolution  rule $\frac{C\lor v\quad D\lor \lnot v}{C\lor D}$, but also using a \emph{$\forall$-reduction} rule $\frac{C\lor u}{C}$, where all the variables in $C$ must be  on the left of $u$ in $\mathcal Q$. 
The \emph{size} of a $\qures$ refutation $\pi$ ($\size(\pi)$) is the number of applications of rules in $\pi$, while \emph{Q-size}  ($\qsize(\pi)$) is the number of applications of the $\forall$-reduction rule.

\paragraph{Algebraic proof systems.}
Given a set of Boolean variables $V$, let $\overline V$ be the set of new formal variables  $\overline v$ for $v\in V$. We consider polynomials with rational coefficients and variables in $V\cup \overline V$, i.e.\ polynomials in the ring $\mathbb Q[V\cup\overline V]$. Given a polynomial $p$ and an assignment $\alpha$ of its variables, we denote with $\restr{p}{\alpha}$ the evaluation of $p$ in $\alpha$. 

In this work we encode clauses and CNF formulas into polynomials using  the so-called \emph{twin-variables encoding}.
A clause $C=\bigvee_{v\in P}v\lor \bigvee_{v\in N}\lnot v$ is encoded as the set of polynomials
\[
\enc(C) = \left\{\prod_{v\in P}\overline v\prod_{v\in N}v \right\}\cup \{v^2-v,\ v+\overline v -1 : v\in P\cup N\}\ .
\]
A CNF  $\phi=\bigwedge_{j=1}^m C_j$ is  encoded as a set of polynomials $\enc(\phi)=\bigcup_{j=1}^{m}\enc(C_j)$. The formula $\phi$ is satisfiable if and only if the set of polynomial equalities $\{p=0 : p\in \enc(\phi)\}$ is satisfiable. 

\begin{fact}
\label{fact:evaluation}
	Given a polynomial $r\in \mathbb Q[V\cup \overline V]$, if $r$ evaluates to $0$ over every Boolean assignment satisfying $\phi$ (and setting $v+\overline v$ to $1$), then $r$ is in the ideal generated by $\enc(\phi)$, i.e.\ there are polynomials $q_p$ such that $r=\sum_{p\in \enc(\phi)}q_pp$.
\end{fact}

A refutation of an unsatisfiable CNF  $\phi$ in variables  $V$ in the  system Nullstellensatz, $\ns$,  (resp.\ Sherali-Adams, $\SA$, or Sum-of-Squares, $\sos$) is an algebraic identity $\pi$ of the form
\begin{equation}
 \label{eq:NS-SA-SOS}
    \sum_{p\in \enc(\phi)} q_p p+ q + 1 = 0\ ,
\end{equation}
where all the polynomials $q_p$, $q$ are in $\mathbb{Q}[V\cup\overline V]$, and  for $\ns$ $q$ is identically $0$ (resp.\ for $\SA$ $q$ is a polynomial with non-negative coefficients, and for $\sos$ $q=\sum_{s\in S}s^2$, that is $q$ is a sum of squares).  The \emph{size} of the $\ns/\SA/\sos$ refutation $\pi$, $\size(\pi)$, is the number of monomials (counted with repetition) in $q_p$ and $q$. The \emph{degree} of $\pi$, $\deg(\pi)$, is the maximum degree of any of the polynomials $q_p p$, and $q$.

\begin{rem}[On variations of $\ns$, $\SA$ and $\sos$]
	The proof system $\ns$ has been considered also for polynomials over arbitrary fields \cite{BGIP.01}. In this paper we focus only on polynomials with rational coefficients. The proof systems $\ns$, $\SA$ and $\sos$ have been also studied using a different encoding of CNF formulas: the encoding $\enc'(C)$, which is the same as $\enc(C)$ but with each $\bar v$ variable substituted by $1-v$.  The systems $\ns$, $\SA$ and $\sos$ under the $\enc'$ encoding are exponentially weaker than the corresponding system under the encoding $\enc$ \cite{RezendeLN021}. In this paper we focus only on polynomials using the encoding $\enc$. The proof system $\sos$ has also been studied recently on Boolean variables representing the Boolean values as $\pm 1$ instead of $0/1$ \cite{Sokolov20}, i.e. using instead of the polynomials $v^2-v$ the polynomials $v^2-1$. In this paper we focus only on polynomials using $0/1$-valued variables. 
\end{rem}

Another well-studied algebraic propositional proof system is \emph{Polynomial Calculus} ($\pc$) \cite{CEI.96}. A Polynomial Calculus (over $\mathbb Q$) refutation of the set of polynomials $\enc(\phi)$, for an unsatisfiable CNF formula $\phi$, is a sequence of polynomials showing that $1$ can be derived from $\enc(\phi)$ using the inference rules  $\frac{p\quad q}{p+q}$ for  polynomials $p,q$, and $\frac{p}{v p}$ where $v$ is a variable or $v\in \mathbb Q$. The \emph{degree} and (monomial) \emph{size} of the refutation are respectively the largest degree of a polynomial in it and the number of monomials in it (counted with multiplicity).

In \cite{BBCP20,BeyersdorffHKS24}, the authors showed how to extend the proof system $\pc$ to the QBF context. This resulted in QBF proof system \emph{Q-Polynomial Calculus} ($\qpc$). $\qpc$ refutes a QBF $\mathcal Q.\phi$ analogously to the system $\pc$,  i.e.\ showing that the polynomial $1$ can be derived from the polynomials in $\enc(\phi)$ using the inference rules  $\frac{p\quad q}{p+q}$ for  polynomials $p,q$, and $\frac{p}{v p}$ where $v$ is a variable or $v\in \mathbb Q$, but also using a \emph{$\forall$-reduction} rule $\frac{p}{\restr{p}{u=b}}$, for $b\in \{0,1\}$ where all the variables in $p$ distinct from $u$ must be   left of $u$ in $\mathcal Q$. The $\size$ of a $\qpc$ refutation is the number of monomials in the refutation (counted with repetition), while the $\qsize$ is the number of monomials in the polynomials involved in the $\forall$-reduction steps (again counted with repetition). The reason to study $\qsize$ and not just $\size$ is to factor out the propositional hardness of the principles and focus on genuine QBF hardness (cf.\ \cite{BHP20,Che17}).

\section{Algebraic systems for QBFs}
\label{sec:algQBF}

We introduce new QBF proof systems inspired by  propositional $\ns/\SA/\sos$. We call them $\qns/\qsa/\qsos$. As their propositional counterparts they are static proof systems: a refutation of a false QBF  $\mathcal Q.\phi$ over variables $V$ in Q-Nullstellensatz, $\qns$, (resp.\ Q-Sherali-Adams, $\qsa$,  and Q-Sum-of-Squares, $\qsos$) is an algebraic identity $\pi$ of the form%
\begin{equation}
 \label{eq:QNS-QSA-QSOS}
    \sum_{p\in \enc(\phi)} q_p p+ \sum_{u\in \mathrm{vars}_\forall(\mathcal Q)}q_u(1-2u)+ q + 1 = 0\ ,
\end{equation}
where all the polynomials $q_p$, $q_u$ $q$ are in $\mathbb{Q}[V\cup\overline V]$, the variables in $q_u$ are all quantified before  $u$ in $\mathcal Q$ (i.e.\ on the left of $u$), and for $\qns$ $q$ is identically $0$ (resp.\ for $\qsa$ $q$ is a polynomial with non-negative coefficients, and for $\qsos$ $q=\sum_{s\in S}s^2$, that is $q$ is a sum of squares). We call the expression in eq.~\eqref{eq:QNS-QSA-QSOS} a \emph{$\qns$-refutation} of $\mathcal Q.\phi$ (resp.\ $\qsa$-refutation/$\qsos$-refutation).

\begin{defi}[size, degree, $\qsize$ and $\qdeg$]
\label{def:deg-size}
 The \emph{size} of a $\qns/\qsa/\qsos$ refutation~$\pi$ ($\size(\pi)$) is the number of monomials (counted with repetition) in $q_p$, $q_u$ and $q$. The \emph{degree} of $\pi$ ($\deg(\pi)$) is the maximum degree of any of the polynomials $q_pp$, $q_u(1-2u)$, and~$q$. 

The \emph{Q-size} of $\pi$ \emph{($\qsize(\pi)$)} is  defined analogously to the $\size$ but accounts only for the monomials in the polynomials $q_u$. The existential Q-degree of $\pi$ \emph{($\qdeg(\pi)$)}  is the maximum existential degree of any $q_u$, where the existential degree is the highest number of existentially quantified variables in any monomial. 	
\end{defi}

The definitions of size and degree for $\qns/\qsa/\qsos$ are completely analogous to the definitions in the propositional setting, while Q-size and Q-degree factor out   propositional hardness and therefore give measures more appropriate to study principles where the hardness stems from quantification.
The definition of Q-size also aligns with genuine QBF hardness measures defined in \cite{BHP20} and analysed e.g.\ in \cite{BBMP22,BeyersdorffHKS24} for QU-Resolution and $\qpc$, where only universal reduction steps are counted. In a sense, the polynomial $q_u$ in \eqref{eq:QNS-QSA-QSOS} can be understood as a universal reduction step on $u$.  In particular, on QBFs without universal variables, $\qns/\qsa/\qsos$ are equivalent to their propositional counterparts $\ns/\SA/\sos$.

Any lower bound on Q-size immediately implies the same lower bound on size.
The reason to consider the \emph{existential} Q-degree is a connection between Q-size and existential Q-degree similar to the inequality between size and width in resolution \cite{BBMP22} (see \Cref{sec:size-degree}).

As a first result we prove that \qns, \qsa, \qsos\ are sound QBF proof systems.
\begin{thm}[soundness]
	\label{thm:soundness}
        If there exists a \qns- or \qsa- or \qsos-refutation of $\mathcal Q.\phi$, then $\mathcal Q.\phi$ is false.
\end{thm}
\begin{proof}
	Suppose, for a contradiction, that $\mathcal Q.\phi$ is a true QBF, so the  $\exists$-player has a winning strategy $\sigma$, but at the same time there is a refutation of $\mathcal Q.\phi$ of the form as in eq.~\eqref{eq:QNS-QSA-QSOS}
where all the variables in polynomials $q_u$ are on the left of $u$ and $q$ is identically zero (\qns), or a polynomial with non-negative coefficients (\qsa), or a sum of squares (\qsos).
	
	For every strategy $\tau$ of the $\forall$-player, the game proceeds following the strategies $\sigma$ and $\tau$ and constructs a total Boolean assignment $\alpha_{\sigma,\tau}$  satisfying the matrix $\phi$. That is $\sum_{p\in \enc(\phi)} q_p p$ evaluates to $0$ under every assignment $\alpha_{\sigma,\tau}$. For a universal variable $u$, we write $\tau_{<u}$ for the part of $\tau$ on variables to the left of $u$ and $\tau_{\ge u}$ for the rest of~$\tau$.	
Taking a uniform probability distribution over all universal strategies $\tau$, for every universal variable $u$ it holds that:
	\[ 
	\E_\tau \left[ q_u (1-2u) \right] \stackrel{(\star)}{=} \E_{\tau_{<u}} \left[ q_u \E_{\tau_{\ge u}} \left[1-2u \right] \right] = \E_{\tau_{<u}} \left[ q_u \cdot 0 \right] = 0\ ,
	\]
		where in the equality $(\star)$ we used the fact that all the polynomials $q_u$ only depend on variables on the left of $u$.
	Hence, evaluating both sides of eq.~\eqref{eq:QNS-QSA-QSOS} on $\alpha_{\sigma,\tau}$ and taking $\E_\tau$, the LHS equals $\E_\tau [\restr{q}{\alpha_{\sigma,\tau}} + 1 ]$, which is always at least $1$, while the RHS is $0$. Contradiction.
\end{proof}

\begin{rem}[\qns over arbitrary fields]
The definition of $\qns$ from eq.~\eqref{eq:QNS-QSA-QSOS} can be trivially adapted from polynomials over $\mathbb Q$ to arbitrary fields of characteristic different from $2$. In characteristic $2$ it gives an unsound system since all the terms $(1-2u)$ are identically~1.  Indeed, in characteristic 2 every formula with at least one universal variable could be ``refuted'' by setting all $q_p=0$ and a single $q_u=1$.
\end{rem}

\begin{rem}[unary vs binary coefficients]
\label{rmk:unary}
Unlike in the propositional setting where the unary versions of $\ns/\SA/\sos$ give rise to non-trivial (and interesting) proof systems \cite{GoosHJMPRT24}, in the QBF setting imposing unary (i.e.\ $\pm 1$) coefficients in $\qns/\qsa/\qsos$ refutations seems to give rise to very weak systems. For instance, unary $\qsos$ cannot even  efficiently refute a false QBF formula as simple as  $\forall u_1 \forall u_2 \cdots \forall u_n. \bigvee_{i=1}^n  u_i$. 
\end{rem}

\begin{prop}
\label{prop:unary}
Any \emph{unary} $\qsos$ refutation of $\forall u_1 \forall u_2 \cdots \forall u_n. \bigvee_{i=1}^n  u_i$  has an exponential number of monomials.
\end{prop}
\begin{proof}
Consider a \qsos-refutation of  $\forall u_1 \forall u_2 \cdots \forall u_n. \bigvee_{i=1}^n  u_i$:
	\begin{equation}
	\label{eq:unary-remark}
    \sum_{p\in \enc(\phi)} q_p p+ \sum_{u\in \mathrm{vars}_\forall(\mathcal Q)}q_u(1-2u)+ q + 1 = 0\ .
\end{equation}
We show that, if the polynomials are written as sums of monomials with $\pm 1$ coefficients then there must be an exponential number of them. The argument is very similar to the one used to prove the soundness of \qsos (\Cref{thm:soundness}).

For every universal strategy $\tau$, the evaluation game following $\tau$ gives a total Boolean assignment $\alpha_{\tau}$. The sum $\sum_{p\in \enc(\phi)} q_p p$ under $\alpha_\tau$ evaluates to $q_0(\alpha_\tau) \prod_{i=1}^n\overline{u_i}(\alpha_\tau)$, which is $0$ for every $\tau$ assigning some $u_i=1$ (and hence $\overline{u_i}=0$). Let $\tau^*$ be the universal strategy assigning all variables to $0$.
	For a universal strategy $\tau$ and universal variable $u$, we write $\tau_{<u}$ for the part of $\tau$ on variables to the left of $u$ and $\tau_{\ge u}$ for the rest of $\tau$.
Taking the uniform probability distribution over all universal strategies $\tau$, for every universal variable $u$ it holds that:
	\[ 
	\E_\tau \left[ q_u (1-2u) \right] \stackrel{(\star)}{=} \E_{\tau_{<u}} \left[ q_u \E_{\tau_{\ge u}} \left[1-2u \right] \right] = \E_{\tau_{<u}} \left[ q_u \cdot 0 \right] = 0\ ,
	\]
	where in the equality $(\star)$ we used the fact that all the polynomials $q_u$ only depend on variables on the left of $u$.
	Evaluating both sides of eq.~\eqref{eq:unary-remark} on $\alpha_\tau$ and taking $\E_\tau$ we obtain that
	\[
	\E_\tau[\restr{q_0}{\alpha_\tau}\prod_{i=1}^n \restr{\overline{u_i}}{\alpha_\tau}+\restr{q}{\alpha_\tau}+1]=
	\frac{1}{2^n}\restr{q_0}{\alpha_{\tau^*}}+\E_\tau[\restr{q}{\alpha_\tau}]+1 =0 \ . 
	\]
	Since $q$ is always non-negative, $\E_\tau[\restr{q}{\alpha_\tau}]\geq 0$ and hence the absolute value of $\restr{q_0}{\alpha_{\tau^*}}$ is at least $2^n$. Therefore the polynomial $q_0$ when written as a sum of monomials with $\pm 1$ coefficients must have an exponential number of them.
	\end{proof}

\subsection{Completeness via a score game}
\label{sec:completeness}

To show  completeness of \qns/\qsa/\qsos, we introduce a new \emph{score} game. We call it \emph{score} game to distinguish it from the \emph{evaluation} game used for the  QBF semantics (cf.\ Sec.~\ref{sec:preliminaries}).

The \emph{score} game, as the evaluation game, is played between a universal and an existential player on a QBF $\mathcal Q.\phi$, building a total Boolean assignment. As in the evaluation game, the players take turns according to the quantifier prefix $\mathcal Q$ and the existential player can freely decide on the value of existential variables.
For the universal variables the score game differs from the usual  evaluation game: the universal player gives a preference for the universal variable $u$ in the form of a number $s_u\in \mathbb Q$. Then, the existential player sets $u$ to $b\in \{0,1\}$ and the universal player scores $s_u(2b-1)$ points. There are two variants of this game that differ in the winning condition:
\begin{description}
\item[variant 1] the universal player wins if $\phi$ is falsified or the total score is strictly positive;
\item[variant 2] the universal player wins if $\phi$ is falsified or the total score equals 1.
\end{description}

Clearly every winning strategy of \textsc{variant 2} is also a winning strategy of \textsc{variant 1}.
The intuition behind the universal preferences is that the sign of $s_u$ encodes the preferred assignment (if the preferred assignment is $u=0$ then the universal player sets $s_u>0$, and $s_u<0$ for $u=1$) and the absolute value encodes the magnitude of this preferred choice. 
If the existential player follows the choice, the universal player loses $|s_u|$ points; otherwise he gains the same amount. 

Our interest in the score game is that the universal winning strategies can be transformed into \qns/\qsos refutations.

\begin{prop}
	\label{prop:game_correctness}
	A QBF $\mathcal Q.\phi$ is false if and only if the universal player has a winning strategy in the score game for $\mathcal Q.\phi$ (in either variant).
\end{prop}
\begin{proof}
	Let the QBF $\mathcal Q.\phi$ be false.
Then there exists a winning strategy $\tau=(\tau_u)_{u\in \vars_{\forall}(\mathcal Q)}$ for the universal player in the evaluation game. 
	In the score game, on the universal variable~$u$ the universal player plays depending on the total score $S$ up to this point and $\tau$. The preferred choice of the universal player is $\tau_u$ and he assigns to $u$ score
	\begin{align*}
		s_u = \begin{cases}
			0 \quad&\text{if }  S = 1\ ,\\
			 (1-2\tau_u)(1-S)\quad&\text{otherwise}\ .
			\end{cases}
	\end{align*}
	
	If in each step of the game the universal player gets always his preference $\tau_u$, then the matrix $\phi$ is falsified (since $\tau$ is a winning strategy in the evaluation game).
	Otherwise, let $u^*$ be the first variable where the universal player does not get his preference and $S^*$ the total score before deciding the value for $u$. Since, by assumption the universal player does not get his preference but the value $1-\tau_{u^*}$ instead, then, after setting $u^*$, the total score is
	\begin{equation}
		S^*+ s_{u^*}(2(1-\tau_{u^*})-1)=1\ .
	\end{equation}
	At this moment the universal player has essentially just won since he can set all following scores to $0$ and the final total score of the game will be $1$.

For the other direction, if the universal player has a winning strategy in the score game on $\mathcal Q.\phi$ then he wins also against the case when the existential player makes the scores negative in each moment of the game. In this case the resulting assignment must falsify the matrix $\phi$ since the universal player is using a winning strategy. As such, this strategy is also a winning strategy for the universal player in the usual evaluation game and the QBF $\mathcal Q.\phi$ is false.
\end{proof}

We require the universal strategy for each universal variable $u$ to be expressed as a polynomial in all variables to the left of $u$ in the quantifier prefix. 
The \emph{size} of a universal strategy in the score game is then the sum of the number of monomials in all the $s_u$.

\begin{thm}
	\label{thm:game_equivalence}
	Let $\pi$ be a shortest \qsos (resp.\ $\qns$) refutation of $\mathcal Q.\phi$ with respect to its Q-size. Then  $\qsize(\pi)$ equals the size of the shortest universal winning strategy in the score game in \textsc{variant~1} (resp.\ \textsc{variant~2}) on $\mathcal Q.\phi$.
\end{thm}
\begin{proof}
Let $U = \vars_{\forall}(\mathcal Q)$, let $S$ be the size of the shortest universal winning strategy in the score game on $\mathcal Q.\phi$, and let $\pi$ be the LHS of a shortest \qsos(resp.\ \qns) refutation written as
\begin{equation}
	\sum_{p\in \enc(\phi)} q_p p+ \sum_{u\in \mathrm{vars}_\forall(\mathcal Q)}q_u(1-2u)+ q = -1\ .
\end{equation} 

To show that $\qsize(\pi)\geq S$ consider the universal strategy setting $s_u=q_u$. This is a winning strategy in the score game, i.e.\ for every total assignment $\alpha$ the universal player wins the score game. Indeed, if $\alpha$ falsifies $\phi$ then the universal player wins automatically (in both versions of the game). Assume then that $\alpha$ satisfies $\phi$, that is for every $p\in \enc(\phi)$, $\restr{p}{\alpha}=0$.
	Therefore, $\pi$ evaluated at $\alpha$ is the same as $\sum_{u \in U} q_u(1-2u) +q$ evaluated at $\alpha$. The expression $\sum_{u \in U} q_u(1-2u)$ evaluates under $\alpha$ to $-1$ if $\pi$ is a $\qns$ refutation or to  $\leq -1< 0$ if $\pi$ is a  $\qsos$ refutation.
	Therefore the total score  when playing the game given the total assignment $\alpha$ is 
	\[
	\sum_{u\in U}s_u(2u-1) = \restr{\left(\sum_{u\in U}q_u(2u-1)\right)}{\alpha} = -\restr{\left(\sum_{u\in U}q_u(1-2u)\right)}{\alpha}
	\] 
	and this latter sum equals $1$ if $\pi$ is a $\qns$ refutation or it is $>0$ if $\pi$ is a $\qsos$ refutation. In other words the universal player in such cases wins using the scores.

	To prove $\qsize(\pi)\leq S$, we consider  the cases where $\pi$ is a \qns or \qsos refutation. 
	
	\medskip
\noindent \emph{Case 1: $\pi$ is a $\qns$ refutation.}
	Let  $(s_u)_{u\in U}$ be a shortest universal winning strategy for the score game in \textsc{variant 2} on $\mathcal Q.\phi$ and let  $q_u$ be the polynomial computing $s_u$ as a function of the variables left of $u$ in $\mathcal Q$.
	In particular, on all Boolean assignments $\alpha$ satisfying the matrix~$\phi$, the universal player wins because the total score is $1$, i.e.\ $\restr{\left(\sum_{u \in U} q_u(1-2u) \right)}{\alpha}=-1$ and therefore $\sum_{u \in U} q_u(1-2u)+1$ is in the ideal generated by the polynomials in $\enc(\phi)$ (this follows from \Cref{fact:evaluation}). 
This  gives a $\qns$ refutation of~$\mathcal Q.\phi$ with a Q-size of at most $S$.        

\medskip
\noindent \emph{Case 2: $\pi$ is a $\qsos$ refutation.} 
	Let  $(s_u)_{u\in U}$ be a shortest universal winning strategy for the score game in \textsc{variant 2}  on $\mathcal Q.\phi$ and let $q_u$ be the polynomial computing $s_u$ as a function of the variables left of $u$ in $\mathcal Q$.
For an assignment $\alpha$, let  $\score(\alpha) = \sum_{u \in U}s_u\restr{(2u-1)}{\alpha}$. For every Boolean assignment $\alpha$ satisfying the matrix $\phi$, since $(s_u)_u$ is a winning strategy, we have $\score(\alpha)>0$. That is for $c=\frac{1}{2}\min_{\alpha\models \phi} \score(\alpha)$ we have 
\begin{equation}
\label{eq:score-SOS}
	\sum_{u\in U} \frac{s_u}{c}(1-2u)=-\frac{\score(\alpha)}{c}<-1\ .
\end{equation}
In this way, for every $\alpha$ satisfying $\phi$, $1-\frac{\score(\alpha)}{c}< 0$. Let 
\[
q=-\sum_{\alpha\models \phi}\left(1-\frac{\score(\alpha)}{c}\right)\chi_{\alpha}(\vec v)\ ,
\]
where $\chi_{\alpha}(\vec v)$ is the monomial which evaluates to $1$ if the variables $\vec v$ are set according to $\alpha$, and $0$ on any other Boolean assignment. Modulo the polynomials $v^2-v$, $q$ is a sum of squares, hence to conclude it is enough to notice that the polynomial
\(
	\sum_{u\in U}\frac{1}{c}q_u(1-2u)+q+1
\)
evaluates to $0$ on every assignment satisfying $\enc(\phi)$,  hence it belongs to the ideal generated by the polynomials in $\enc(\phi)$ (this follows from \Cref{fact:evaluation}). This gives a $\qsos$ refutation of $\mathcal Q.\phi$ having the same $\qsize$ as the strategy $(s_u)_u$.
\end{proof}

\begin{cor}[completeness]
	\label{cor:NS_complete}
	\qns/\qsa/\qsos are complete.
\end{cor}
\begin{proof}
	 Proposition~\ref{prop:game_correctness} and Theorem~\ref{thm:game_equivalence}  immediately imply the completeness of \qns. As \qns refutations are  special cases of \qsa and \qsos refutations their completeness also follows.
\end{proof}

Alternatively, it is also possible to prove the completeness of $\qns$ directly (without using the score game), as follows. 
\begin{proof}[Alternative proof of completeness of \qns]
It is sufficient to show the completeness without twin variables. For this proof we use the encoding where $\enc(x)=1-x$, $\enc(\neg x)=x$,
	$\enc(C)=\prod_{\ell\in C}\enc(\ell)$. 
	
	For a partial assignment $\rho$, let $\Ind(\rho)$ be the ``indicator'' polynomial that evaluates to $1$ on all assignments extending $\rho$, and to $0$ on all
	Boolean assignments that disagree with $\rho$. Formally, let
	\[ 
	\Ind(\rho)=\prod_{x\in \rho^{-1}(1)} x\prod_{x\in \rho^{-1}(0)} (1-x)\ .
	\]
	
	Notice that if $\rho$ falsifies $C$, then $\enc(C)$ divides
	$\Ind(\rho)$.
	Given a winning strategy $S$ for the universal player in the
	evaluation game,   we use it to prune the complete assignment tree, i.e.	remove branches not consistent with $S$. 
	
	Furthermore, we terminate a path
	as soon as a clause is falsified, and label the corresponding leaf with that falsified clause. (If there is more than one, pick any one
	arbitrarily.) 
	This is a decision tree $T$ which, given any
	assignment $\alpha$ to the existential variables, finds a clause
	falsified by the assignment $\alpha \cup S(\alpha)$. 
	The tree is unary-binary, since nodes labelled by universal variables have only one child.
	For a node $v$ in the tree $T$, let $\alpha_v$ denote the partial
	assignment leading to $v$. For the root node $r$,
	$\alpha_r=\emptyset$, and $\Ind(\alpha_r) = 1$.
	
	We show, by a bottom-up traversal, that for each node $v$, the
	polynomial $\Ind(\alpha_v)$ can be expressed as a polynomial
	combination of the initial clause encodings $\{\enc(C)\mid C\in F\} $
	and the universal polynomials $(1-2u)$, where the multipliers for
	$1-2u$ involve only variables quantified left of $u$.

\noindent \textbf{Case 1:} $v$ is a leaf, where clause $C$ is falsified. In this case, 
		$\Ind(\alpha_v)$ is a polynomial multiple of $\enc(C)$.
		
\noindent \textbf{Case 2:} $v$ is a node with two children, $v_0$ and
		$v_1$. Then, for some existential variable $x$, $\alpha_{v_b}$
		extends $\alpha_v$ by $x=b$. We have inductively expressed
		$\Ind(\alpha_{v_b})$ for $b=0,1$. Now
		\begin{align*}
					\Ind(\alpha_v) &= \Ind(\alpha_{v_0}) + \Ind(\alpha_{v_1})\ .
					\\
		\Ind(\alpha_{v_0}) + \Ind(\alpha_{v_1}) &=
		\Ind(\alpha_v)(1-x) + \Ind(\alpha_v)x 
		\\
		&= \Ind(\alpha)v\ .
		\end{align*}
		
\noindent \textbf{Case 3:} $v$ is a node with one child $w$. Then for some
		universal variable $u$ and some $b$, $\alpha_w$ extends $\alpha_v$
		by $u=b$, and we have inductively expressed $\Ind(\alpha_w)$.
		Consider $b=0$, then 
		\begin{align*}
			2\,\Ind(\alpha_w) - \Ind(\alpha_v)(1-2u) &=
		2\, \Ind(\alpha_v)(1-u) - \Ind(\alpha_v)(1-2u) \\
		&=
		\Ind(\alpha_v)\ .
		\end{align*}
		Since all variables in $\alpha_v$ are quantified
		left of $u$, $\Ind(\alpha_v)$ is a valid multiplier for $(1-2u)$.
		Similarly if $b=1$, then 
		\begin{align*}
					2\, \Ind(\alpha_w) + \Ind(\alpha_v)(1-2u) & =
		2\,  \Ind(\alpha_v)u + \Ind(\alpha_v)(1-2u) 
		\\
		&=
		\Ind(\alpha_v)\ .
		\end{align*}
	
	Proceeding in this way all the way to the root, we get that
	$\Ind(\alpha_r)=1$ is expressible as a polynomial combination of the
	clause encodings and the universal polynomials with valid multipliers.
\end{proof}

\subsection{Upper bounds in \qsos via the score game}
Due to \Cref{thm:game_equivalence}, the score game can be used to obtain bounds on the Q-size of $\qsos$ refutations. 
The advantage is being able to argue directly on countermodels (without reference to the syntactic representation of the matrix).
To that end, we use QBFs $Q\text{-}C_n$  from \cite{BBCP20}, which are defined via their countermodel computed by a family of circuits $C_n$.

\begin{defi}[$Q\text{-}C_n$ \cite{BBCP20}]
	Let $n$ be an integer and $C_n$ be a circuit with inputs $x_1, \dots, x_n$ and a single output. 
	We define
	\begin{equation*}
		Q\text{-}C_n = \exists x_1 \cdots \exists x_n \forall u \exists t_1 \cdots \exists t_m. \big(u \nleftrightarrow C_n(x_1, \dots, x_n)\big)\ ,
	\end{equation*}
	where the additional variables $t_i$ are used for a Tseitin-encoding of the circuit $C_n$ into  CNF (the $i^\text{th}$ node in $C_n$ is represented by  variable $t_i$,  and if, e.g.,  node $i$ is an $\land$-gate between nodes $j$ and $k$, then we have clauses encoding $t_i\leftrightarrow( t_j\land t_k)$, and similarly for $\lor$ and $\neg$ gates).
\end{defi}

For the $Q\text{-}C_n$ formulas, the only countermodel sets $u$ to  $C_n(x_1, ..., x_n)$.
Here, we are specifically interested in choosing circuits $C_n$ that  compute $\maj_n$. 
A circuit calculating $\maj_n$ evaluates to true, if and only if at least half of the $n$ input variables are set to true.
We show that $\qmaj_n$ has short \qsos refutations in the $\qsize$ measure.

\begin{prop}
\label{prop:majority}
	$\qmaj_n$ has \qsos refutations of linear $\qsize$.
\end{prop}
\begin{proof}
	In the score game for $\qmaj_n$, set $s_u = - x_1 - \cdots - x_n + \frac{n}{2} - \frac14$. The choice of the constant $ \frac14$ is somewhat arbitrary, but should be between 0 and  $\frac12$.
	As there is only a single universal variable, this defines a complete strategy for the universal player. 
	We show that $s_u$ is a winning strategy.
	For an arbitrary assignment~$\alpha$, 	$\restr{s_u}{\alpha} < 0$ if and only if at least half of the existential $x_i$ variables are set to 1, otherwise $\restr{s_u}{\alpha} > 0$. As such, the total score of the score game on $\alpha$ is
	\(
		 \restr{\big(-x_1 - \cdots - x_n + \frac{n}{2} - \frac14\big)(2u-1)}{\alpha}
	\).
	This is negative only if either $u=1$ and at least half of the $x_i$ equal 1 or $u=0$ and less then half of the $x_i$ equal 1,
	i.e.\ if the matrix is satisfied, the score is positive.    
	
	This strategy has size $n+1$ and degree $1$, hence, by \Cref{thm:game_equivalence}, there exists a \qsos refutation of $\qmaj_n$ with a $\qsize$ linear in $n$.
\end{proof}

It is known that $\qmaj_n$ requires $\qpc$ refutations of exponential $\qsize$ \cite{BeyersdorffHKS24}, therefore the previous result yields the exponential separation between \qpc and \qsos  in \Cref{fig:p-sim}.

\subsection{From existential Q-degree to Q-size}
\label{sec:size-degree}

In various proof systems, strong enough lower bounds on the degree/width of proofs immediately imply non-trivial lower bounds on proof size. This happens for instance in  propositional proof systems such as Resolution \cite{BW01}, Polynomial Calculus \cite{CEI.96}, Sherali-Adams and Sum-of-Squares \cite{AH.19}; and in QBF proof systems as well, for instance in QU-Resolution \cite{BBMP22}, and Q-PC \cite{BeyersdorffHKS24}.
It turns out that a very similar statement holds in \qsos between the existential Q-degree ($\qdeg(\cdot)$, see \Cref{def:deg-size}) and Q-size ($\qsize(\cdot)$, see \Cref{def:deg-size}).

\begin{thm} \label{thm:sizedegree}
	Let $Q. \varphi$ be a false QBF with $n$ variables  that has a \qsos refutation of $\qsize$ $s$. Then it has a \qsos refutation of $\qdeg$ $O(\sqrt{n \log s})$.
\end{thm}

The argument is similar to the proof of the analogous  size-width inequality for Resolution from \cite{BW01}. The main difference is the proof of the lemma below showing how to combine a proof of $\qdeg$ $k-1$ of $\mathcal Q.\phi|_{x=1}$ and a proof of $\qdeg$ $k$ of $\mathcal Q.\phi|_{x=0}$ into a proof of $\qdeg$ $k$ of $\mathcal Q.\phi$. This is done using the score game from the previous section.

\begin{lem} \label{lemma:sizedegree_helper}
	Let $Q. \varphi$ be a false QBF and $x\in \mathrm{vars}_\exists(\mathcal Q)$. If there is a $\qsos$ refutation  $\pi_1$ of $Q. \varphi|_{x=1}$ with  $\qdeg(\pi_1)\leq k-1$, and a $\qsos$ refutation $\pi_0$ of $Q. \varphi|_{x=0}$ with $\qdeg(\pi_0)\leq k$, then there is a $\qsos$ refutation $\pi$ of $\mathcal Q.\phi$ with $\qdeg(\pi)\leq k$.
\end{lem}

\begin{proof}
	We consider the equivalent representation of the proofs as universal strategies in the score game.
	Let $q_1$ be the final score in $\pi_1$ and $q_0$ be the final score in $\pi_0$.
	Let $c$ be the maximum of $\abs{q_0} + 1$ over all assignments, and $d$ be the smallest \emph{positive} value that $q_1$ can take (or $1$ if $q_1$ does not take positive values).
	For each $u\in \mathrm{vars}_\forall(\mathcal Q)$, let  $q_u = x \cdot q_u^1 + \frac{d}{c} q_u^0$, where $q_u^0$ and $q_u^1$ are the polynomials for $u$ in $\pi_0$ and $\pi_1$. Combining these $q_u$ yields a strategy~$\pi$ with final score $q = x \cdot q_1 + \frac{d}{c} q_0$ and $\qdeg(\pi)\leq k$.

	We still need to argue that $\pi$ is a universal winning strategy, i.e.\ that on every assignment~$\alpha$ satisfying $\varphi$ we have $q(\alpha) > 0$.
	If $\alpha$ sets $x=0$, then it satisfies $\varphi|_{x=0}$, so $q_0(\alpha) > 0$ due to the correctness of $\pi_0$. But if $x=0$ then $q = 0 \cdot q_1 + \frac{d}{c} q_0 > 0$.
	If $\alpha$ sets $x=1$, then it satisfies $\varphi|_{x=1}$, so $q_1(\alpha) > 0$ due to the correctness of $\pi_1$. By the definitions of $c$ and $d$, we have $c > -q_0(\alpha)$ and $d \le q_1(\alpha)$. This means that $\frac{d}{c} q_0 > -d \ge -q_1(\alpha)$ and $q = 1 \cdot q_1 + \frac{d}{c} q_0 > 0$.
\end{proof}

\begin{restatable}{lem}{degreelemma}
\label{lem:sizedegree_general}
	Let $d,n,b \in \mathbb{N}^{\ge 0}$ and $Q.\varphi$ be a false QBF. Let $\pi$ be a $\qsos$ refutation of $Q.\varphi$ so that its $q_u$ polynomials contain, in total, fewer than $(1-\frac{d}{2n})^{-b}$ monomials of existential degree $> d$. Then there is a $\qsos$ refutation $\pi'$ of $Q.\varphi$ with $\qdeg(\pi')\le d+b$.
\end{restatable}

\begin{proof}
The proof of this lemma is virtually identical to the proof of \cite[Theorem 3.5]{BW01}. 
Informally, the proof is by an inductive argument on $n$ and $b$, considering the \emph{high-degree} monomials to be the ones with $\qdeg$ at least $d$. By a counting argument there will be a literal $x$ appearing in at least a $\frac{d}{2n}$ fraction of them. Restricting by $x=0$ and $x=1$, we have that the first restriction eliminates at least $\frac{d}{2n}$ of the high-degree monomials, while the second eliminates one variable. Then using the inductive hypothesis and \Cref{lemma:sizedegree_helper} concludes the argument.

More formally, we fix $d$ and prove the statement by induction over $n$ and $b$.
	In the base cases where $b=0$ or $n \le d$, $\pi' = \pi$ has the required $\qdeg$.
	Otherwise, we call those monomials \emph{high-degree} which have existential degree $> d$. If $\pi$ has $k < (1-\frac{d}{2n})^{-b}$ high-degree monomials, they contain in total $> kd$ existential literals. Because there are only $2n$ existential literals, one of these, say $x$, must occur in more than $\frac{kd}{2n}$ high-degree monomials.

	When restricting $\pi$ to $x=0$, all these monomials vanish, leaving fewer than $k - \frac{kd}{2n} = k(1 - \frac{d}{2n}) < (1 - \frac{d}{2n})^{-(b-1)} < (1 - \frac{d}{2(n-1)})^{-(b-1)}$ high-degree monomials.
	Applying the induction hypothesis yields a proof $\pi_0$ of $Q.\varphi|_{x=0}$ with $\qdeg$ $\le d+b-1$.
	When restricting $\pi$ to $x=1$, the number of high-degree monomials certainly cannot increase, so it is at most $k < (1-\frac{d}{2n})^{-b} < (1-\frac{d}{2(n-1)})^{-b}$. Applying the induction hypothesis yields a proof $\pi_1$ of $Q.\varphi|_{x=1}$ with $\qdeg$ $\le d+b$.

	Finally, applying Lemma~\ref{lemma:sizedegree_helper} to $\pi_0$ and $\pi_1$ yields a proof $\pi'$ of $\qdeg$ $\le d+b$.
\end{proof}

Given the two lemmas above it is immediate to prove \Cref{thm:sizedegree}.

\begin{proof}[Proof of \Cref{thm:sizedegree}]
	Set $b = d = \sqrt{2n \log s}$ and observe that $s < (1-\frac{d}{2n})^{-b}$.
	The number of high-degree monomials in the refutation is  smaller than its total number of monomials $s$, so we can apply Lemma~\ref{lem:sizedegree_general} and get a refutation of $\qdeg(b+d) \in O(\sqrt{n \log s})$.
\end{proof}

\section{Lower bounds via strategy extraction in the evaluation game}
\label{sec:strategy}

In QBF proof systems, strategy extraction is a welcome and ubiquitous feature. 
Informally, given a refutation of a false QBF, strategy extraction allows  to represent in some computational model a winning strategy of the universal player  in the \emph{evaluation} game. 
Different QBF proof systems give rise to strategy extraction in different computational models.
For example, from \qures refutations we get \emph{unified decision lists} \cite{BBMP22} and from Frege+$\forall$-reduction refutations we get $\mathsf{NC}_1$ circuits \cite{BBCP20} .

In this section, we show that \qsos/\qsa/\qns also admit strategy extraction, using \emph{polynomial threshold functions} (PTF) as the computational model (\Cref{thm:strategy-extraction}). We use this fact to prove a lower bound in \qsos (\Cref{cor:parity}) and a p-simulation of \qsos by $\qtcz$ (\Cref{cor:tcz}).
\begin{defi}[polynomial threshold function]
A Boolean function $f:\{0,1\}^n \to \{\pm 1\}$ is computed by a \emph{polynomial threshold function (PTF)}, if there exists some $n$-variate polynomial $p\in \mathbb Q[\vec x]$ such that $f(\vec x) = \sign(p(\vec x))$. 
The \emph{size} of the PTF is the number of monomials in $p$ and the \emph{degree} of the PTF is the degree of $p$.
\end{defi}

\begin{rem}[On the PTF degree of parity]
\label{rem:PTF-parity}
It is well known that $f(\vec x)=x_1\oplus x_2\oplus \cdots\oplus x_n$ cannot be computed by PTFs of degree less than $n$. We recall briefly the argument (see also for instance \cite{MP87}). Let $f(\vec x)=\sign(p(\vec x))$ with $p$ a $n$-variate polynomial of degree $d$. Since $f$ is symmetric, there exists  a symmetric $n$ variate polynomial $p'$ such that $f(\vec x)=\sign(p'(\vec x))$ and $\deg(p')\leq \deg(p)=d$. Since $p'$ is symmetric, there exists a univariate polynomial $p''$ such that $p''(x_1+\cdots+x_n)=p'(\vec x)$ and $\deg(p'')=\deg(p')$. In other words, $f(\vec x)=\sign(p''(x_1+\dots+x_n))$. To conclude it is enough to notice that on $0$, $1$, $\dots$, $n$ the polynomial $p''$ must be alternating signs and hence it must have at least $n$ real roots; therefore $d\geq \deg(p'')\geq n$.
\end{rem}

Given a \qsos/\qsa/\qns refutation 
\[ \sum_{p\in \enc(\phi)} q_p p+ \sum_{u\in \mathrm{vars}_\forall(\mathcal Q)}q_u(1-2u)+ q + 1=0 \]
of a false QBF $\mathcal Q.\phi$, we claim that for every universal variable $u$, $\sign(q_u)$ computes a strategy for $1-2u$.\footnote{We use the convention that $\sign(0)=+1$.}
This can then be easily transformed into a strategy for $u$ through a linear output transformation mapping $-1$ to $1$ and $1$ to $0$. 
Per definition, for every universal variable $u$, $\sign(q_u)$ is a PTF. The size of the extracted strategy is the sum of the sizes of the PTFs and, as such, the Q-size of the refutation. 
Analogously, the degree of the extracted strategy, i.e. the maximum degree of the PTFs, equals the total Q-degree of the refutation.

\begin{thm}
\label{thm:strategy-extraction}
	  Let $ \sum_{p\in \enc(\phi)} q_p p+ \sum_{u\in \mathrm{vars}_\forall(\mathcal Q)}q_u(1-2u)+ q + 1 =0$ be a \qsos / \qsa / \qns proof of a false QBF $\mathcal Q.\phi$. Then the universal strategy that maps each universal variable $u\in \vars_{\forall}(\mathcal Q)$ to
	  \(
	 \frac{1}{2}(1-\sign(q_u))
	  \)
	  is a countermodel (i.e. falsifies $\phi$).
\end{thm}
\begin{proof}
	The syntactic restriction of countermodels, i.e.\  that each variable  only depends on  variables left of it in the quantifier prefix $\mathcal Q$, holds by  definition of $q_u$.
	
	For every universal variable $u$ played according to the strategy, we have $q_u(1-2u)=q_u\sign(q_u) \ge 0$.
	As such, 
	$	\sum_{u \in \vars_{\forall}(\mathcal Q)} q_u(1-2u) + q + 1 \ge 1$.
	Hence, $\sum_{p \in \enc(\phi)} q_pp \le -1$, which is only possible if the matrix $\phi$  is not satisfied (otherwise $q_pp=0$ for all $p \in \enc(\phi)$).
\end{proof}

To exemplify the strategy extraction technique for \qsos, we use  \cref{rem:PTF-parity} and \cref{thm:strategy-extraction} to prove that the \parityformula formulas \cite{BCJ15}
\begin{equation*}
	\parityformula_n=\exists x_1 \cdots\exists x_n \forall u \exists t_1 \cdots \exists t_n. \ (t_1 \leftrightarrow x_1) \land (u \nleftrightarrow t_n) \land \bigwedge_{i=2}^n (t_i \leftrightarrow t_{i-1} \oplus x_i).
\end{equation*}
are exponentially hard for \qsos.
Notice that the only winning strategy for the universal player is to set $u=x_1\oplus\dots\oplus x_n$.

\begin{cor}
\label{cor:parity}
	Every \qsos refutation of $\parityformula_n$ requires Q-size $\exp(\Omega(n))$.
\end{cor}

\begin{proof}
  	Let $d := \qdeg(\pi)$ and apply strategy extraction (\cref{thm:strategy-extraction}) to get a PTF of degree $d$ that computes $u = \bigoplus_{i=1}^n x_i$. By \Cref{rem:PTF-parity}, its degree is at least $n$, so $d \ge n$. Let $s := \qsize(\pi)$ and apply Theorem~\ref{thm:sizedegree} to obtain $n \le d = O(\sqrt{n \log s})$ and therefore $s = \exp(\Omega(n))$.
	Theorem~\ref{thm:sizedegree} can be applied here, because $\parityformula_n$ only has a single universal variable and, as such, its existential Q-degree equals its total Q-degree.
\end{proof}

As a second consequence of \cref{thm:strategy-extraction},  strategy extraction can also be used to embed \qsos into more powerful systems, in this case \qtcz. \qtcz is the $\fregeTC$ system with an added universal reduction rule. The Q-size of a \qtcz refutation is the sum of the number of symbols of all lines involved in a  $\forall$-reduction step. 

\begin{restatable}{cor}{cortcz}
	\label{cor:tcz}
	$\qtcz$ p-simulates \qsos in the $\qsize$ measure.
\end{restatable}

\begin{proof}
	We  extract PTFs calculating a countermodel for the QBF from {\qsos} refutations. 
	A PTF can be transformed into a depth-2 threshold circuit as follows.
	Every monomial can be seen as an AND Gate, since we are only interested in Boolean assignments. In the second layer, we can use a single weighted linear threshold function using the AND gates as inputs.
	
	It is known that the $\qsize$ of {\qtcz} refutations is characterized by the size of expressing the strategy in $\mathsf{TC}_0$ circuits \cite{BBCP20}.
	As such, there exists a \qtcz proof of $\qsize$ bounded by the size of the $\mathsf{TC}_0$ circuit computing the countermodel, which itself is bounded by the $\qsize$ of the $\qsos$ refutation.
\end{proof}

\section{Lower bounds via Q-pseudo-expectation}
\label{sec:pseudoexpectation}

In the propositional setting, the notion of pseudo-expectation is the standard tool to obtain degree lower bounds for $\sos$, see for instance \cite{FKP.19}. Thanks to the size-degree relation for $\sos$ \cite{AH.19},  degree lower bounds also give  size lower bounds.

Inspired by the propositional notion of pseudo-expectation, we give a notion of pseudo-expectation for $\qsos$ and  use it to prove lower bounds on the $\qdeg$ of $\qsos$ refutations. In particular, for a false QBF $\mathcal Q.\phi$ and a $\qsos$ expression $\pi$ of the form
\begin{equation}
\label{eq:pseudo-exp}
	\sum_{p\in \enc(\phi)} q_p p+ \sum_{u\in \mathrm{vars}_\forall(\mathcal Q)}q_u(1-2u)+ q + 1\ ,
\end{equation}
where all the variables in $q_u$ are on the left of $u$ in $\mathcal Q$ and $q$ is a sum of squares, we consider  \emph{witnesses} that $\pi\neq 0$, i.e.\ that $\pi$ \emph{is not} a $\qsos$ refutation of~$\mathcal Q.\phi$. In analogy to the propositional case, we call the witnesses we construct  \emph{Q-pseudo-expectations}. 

\begin{defi}[Q-pseudo-expectation in $\qsos$]
\label{def:Q-pseudoexpectation}
Given $\mathcal Q.\phi$ and a $\qsos$ expression $\pi$  as in eq.~\eqref{eq:pseudo-exp}, a  \emph{Q-pseudo-expectation} for $\mathcal Q.\phi$ and $\pi$ is a linear function $\psE \colon \mathbb Q[V\cup\overline V]\to \mathbb R$  such that:
	\begin{enumerate}
	\item $\psE[1] = 1$;
	\item $\psE[q + \sum_{p\in \enc(\phi)} q_p p]\ge 0$; 
	\item $\psE[\sum_{u \in \vars_{\forall}(\mathcal Q)}q_u(1-2u)] \ge 0$.
\end{enumerate}
\end{defi}

In the propositional context, a single pseudo-expectation for $\sos$ typically targets a fixed degree~$d$ and has properties similar to 1.--3. above but for arbitrary $q_p$ and sum-of-squares~$q$ such that the degree of $q_pp$ and $q$ are at most $d$. In this way, a single pseudo-expectation rules out the possibility of \emph{any} small-degree $\sos$ refutation. There are exceptions to this general approach, for instance the pseudo-expectations used in \cite{Hakoniemi20} that are targeting all $\sos$ proofs over a certain set of monomials, \emph{but} we are not aware of degree lower bounds in $\sos$ proved by constructing a family of pseudo-expectations each tailored to a specific set of polynomials (i.e. as in \cref{def:Q-pseudoexpectation} but without the condition in item 3.). In the QBF context, this is what we do. To rule out small $\qdeg$ $\qsos$ refutations we use a family of pseudo-expectations, each targeting \emph{one} possible candidate $\qsos$ proof, i.e.\ an expression as in eq.~\eqref{eq:pseudo-exp}. We formalise this approach in \Cref{thm:pseudoExp} and exemplify it in \Cref{thm:Equality}.

\begin{thm} \label{thm:pseudoExp}
	Given a QBF $\mathcal Q. \phi$, if for every $\qsos$ expression $\pi$ as in eq.~\eqref{eq:pseudo-exp} with $\qdeg(\pi)<d$ there is a pseudo-expectation for $\mathcal Q. \phi$ and $\pi$, then every $\qsos$ refutation of $\mathcal Q. \phi$ has $\qdeg$ at least $d$.
\end{thm}

\begin{proof}
If there was a $\qsos$ refutation $\pi=0$ of $\mathcal Q.\phi$ with $\qdeg(\pi)<d$, then taking the pseudo-expectation $\psE$ for $\mathcal Q.\phi$ and $\pi$ we get $\psE[\pi]=\psE[0]=0$. Notice that  necessarily $\psE[0]=0$, by linearity and the identity $\psE[0+0]=\psE[0]$. On the other hand, again by linearity, and the properties of $\psE$, we get $\psE[\pi]\geq 1$.
\end{proof}

We apply the Q-pseudo-expectation technique to show a $\qdeg$ lower bound on $\qsos$ refutations of  the   $\equality_n$ formulas \cite{BBH18} where
\[ 
\equality_n=\exists x_1 \cdots x_n \forall u_1 \cdots u_n \exists t_1 \cdots t_n. \bigwedge_{i=1}^n (t_i \rightarrow (x_i \nleftrightarrow  u_i)) \land \bigvee_{i=1}^n t_i\ .
\]

For the $\qdeg$ lower bound we only use the fact that every assignment satisfying the matrix sets $u_i \nleftrightarrow x_i$ for some $i$.

$\equality_n$ is a quite simple QBF. Hardness for $\qsos$ might suggest that the system is  weak. $\equality_n$ is also hard for $\qures$ and $\qpc$ \cite{BBH18}, but easy in $\mathsf{Q\text{-}depth\text{-}}d\ \mathsf{Frege}$ \cite{BBH18}, so  the hardness appears to only stem  from the expressiveness of the objects used. Similarly, $\qsos$ on depth-$d$ arithmetic circuits (instead of polynomials) would  shortly  prove $\equality_n$, but if  polynomials are represented explicitly as sums of monomials,  $\equality_n$ becomes hard.

\begin{thm}
\label{thm:Equality}
	Every $\qsos$ refutation $\pi$ of $\equality_n$ has $\qdeg(\pi)\ge n$ and $\qsize(\pi)\geq \exp(\Omega(n))$.
\end{thm}

\begin{proof}
  Let $\mathcal Q.\phi$ be the QBF encoding of $\equality_n$. First notice that, thanks to \Cref{thm:sizedegree}, it is enough to prove the $\qdeg$ lower bound. (We comment that to do this, the strategy extraction technique from \Cref{sec:strategy}  would not work here.)
  Assume, towards a contradiction, that there is a $\qsos$ refutation of $\mathcal Q.\phi$ with $\qdeg(\pi)<n$:
\begin{equation}
\label{eq:qsos-refutation-equality}
	\sum_{p\in \enc(\phi)} q_p p+ \sum_{u\in \mathrm{vars}_\forall(\mathcal Q)}q_u(1-2u)+ q + 1 =0\ .
\end{equation}

Let $\pi$ the LHS of eq.~\eqref{eq:qsos-refutation-equality}.
We construct a Q-pseudo-expectation for $\mathcal Q.\phi$ and $\pi$. This, by \Cref{thm:pseudoExp}, implies the wanted contradiction. 

Given $\vec \alpha=(\alpha_1,\dots,\alpha_n)\in \{0,1\}^n$, let $\vec x\mapsto\vec \alpha$ be the Boolean assignment setting $x_i$ to $\alpha_i$ for each $i\in [n]$. We define analogously $\vec u\mapsto\vec \alpha$ and $\vec t\mapsto\vec \alpha$.  Given $\vec \alpha,\vec \beta\in \{0,1\}^n$, let $\vec\alpha\oplus\vec\beta$ be the vector whose $i$th entry is the sum $\alpha_i+\beta_i\pmod 2$. 

Let $h = \sum_{u\in \vars_{\forall}(\mathcal Q)} q_{u}(1-2u)$. We have that
\begin{equation}
\label{eq:h-alpha-beta}
	\sum_{\vec \alpha,\vec \beta\in \{0,1\}^n} \restr{h}{\vec x\mapsto\vec\alpha,\ \vec u\mapsto\vec \beta}
	=
	0\ .
\end{equation}
Let $\vec \gamma\in \{0,1\}^n$ be the assignment maximizing $\sum_{\vec \alpha\in \{0,1\}^n} \restr{h}{\vec x\mapsto\vec\alpha,\ \vec u\mapsto\vec\gamma}$. 
By eq.~\eqref{eq:h-alpha-beta}, 
\begin{equation}
\label{eq:h-gamma}
	\sum_{\vec \alpha\in \{0,1\}^n} \restr{h}{\vec x\mapsto\vec\alpha,\ \vec u\mapsto\vec\gamma}\geq 0\ .
\end{equation}
We define our candidate Q-pseudo-expectation as
	\[ 
	\psE(p) = 2^{-n} \sum_{\vec \alpha \in \{0,1\}^n }\restr{p}{\vec x\mapsto\vec \alpha,\ \vec u\mapsto\vec \gamma,\ \vec t\mapsto\vec \alpha\oplus\vec \gamma}\ , 
	\]
	and we prove that $\psE$ satisfies all the conditions of the definition of Q-pseudo-expectation.
	The definition of $\psE$ immediately implies that $\psE[1] = 1$, and $\psE[h] \ge 0$ follows  from eq.~\eqref{eq:h-gamma}.
Let $g = q + \sum_{p\in \enc(\phi)} q_p p$. To conclude the argument we need to prove that $\psE[g]\ge 0$. 

 Since $g=-h-1$ and by construction $h$ cannot contain $t_i$ variables, the polynomial $g$  contains only $x_i$ and $u_i$ variables. That is, 
 \[
 \psE[g]
 =
 2^{-n}\sum_{\vec \alpha \in \{0,1\}^n }\restr{g}{\vec x\mapsto\vec \alpha,\ \vec u\mapsto\vec \gamma}
 \stackrel{(\star)}{=}
 2^{-n}\sum_{\vec \alpha \in \{0,1\}^n }\restr{g}{\vec x\mapsto\vec \alpha\oplus \vec \gamma,\ \vec u\mapsto\vec \gamma}\ ,
 \]
 where the last equality $(\star)$ follows as $\vec \alpha\oplus\vec \gamma$ ranges over all $\{0,1\}^n$ just in a different order.

Let $g'$ be the polynomial such that for every $\vec \alpha\in \{0,1\}^n$, 
\[
\restr{g'}{\vec x\mapsto\vec \alpha}=\restr{g}{\vec x\mapsto\vec \alpha\oplus \vec \gamma,\ \vec u\mapsto\vec \gamma}\ .
\] 
The polynomial $g'$ is constructed from $\restr{g}{\vec u\mapsto\vec\gamma}$ replacing every occurrence of $x_i$ by $1-x_i$ (resp. $x_i$) if $\gamma_i=1$ (resp. $\gamma_i=0$).
	Crucially, $g'$ has degree at most $n-1$, since $\deg(g')\leq \deg_\exists(g) = \deg_\exists(-h-1)\leq \qdeg(\pi)\leq n-1$. That is, for each monomial $m$ in $g'$, there is some variable $x_i$ not appearing in it, and $m$ gets the same value on every pair $\vec \alpha, \vec \alpha'$ where $\vec \alpha\in \{0,1\}^n$ and $\vec \alpha'$  is identical to $\vec\alpha$ except in position $i$. This implies 
	\(
		\sum_{\vec \alpha \in \{0,1\}^n} (-1)^{\abs{\vec \alpha}} \restr{m}{\vec x\mapsto\vec\alpha}=0
	\), 
	where $|\vec \alpha|$ denotes the $1$-norm of $\vec \alpha$, the sum of all $1$s in it.
Summing over all monomials in $g'$, we~get
\begingroup
\allowdisplaybreaks
	\begin{align*}
	0&=\sum_{\vec \alpha \in \{0,1\}^n} (-1)^{\abs{\vec \alpha}} \restr{g'}{\vec x\mapsto\vec\alpha}
	\\
	&=
	\sum_{\vec \alpha \in \{0,1\}^n} \restr{g'}{\vec x\mapsto\vec\alpha}+\sum_{\vec \alpha \in \{0,1\}^n} ((-1)^{\abs{\vec \alpha}}-1) \restr{g'}{\vec x\mapsto\vec\alpha}
	\\
	&=\sum_{\vec\alpha\in \{0,1\}^n}\restr{g}{\vec x\mapsto\vec \alpha\oplus \vec \gamma,\ \vec u\mapsto\vec \gamma}+\sum_{\vec \alpha \in \{0,1\}^n} ((-1)^{\abs{\vec \alpha}}-1) \restr{g'}{\vec x\mapsto\vec\alpha}
	\\
	&=2^n\psE[g]+\sum_{\vec \alpha \in \{0,1\}^n} ((-1)^{\abs{\vec \alpha}}-1) \restr{g'}{\vec x\mapsto\vec\alpha}\ .
	\end{align*}
\endgroup
Hence, to conclude that  $\psE[g]\geq 0$ it is enough to show that 
	\begin{equation}
	\label{eq:sum-1}
			\sum_{\vec \alpha \in \{0,1\}^n} ((-1)^{\abs{\vec \alpha}}-1) \restr{g'}{\vec x\mapsto\vec\alpha}\leq 0\ .
	\end{equation}
For $\vec \alpha$s such that $|\vec \alpha|$ is even, the coefficient in front of $\restr{g'}{\vec x\mapsto\vec \alpha}$ is $0$, while for $\vec\alpha$s such that $|\vec \alpha|$ is odd, the coefficient in front of $\restr{g'}{\vec x\mapsto\vec \alpha}$ is $-2$. That is, to prove eq.~\eqref{eq:sum-1}, it suffices to show that for $\vec\alpha$s such that $|\vec \alpha|$ is odd, $\restr{g'}{\vec x\mapsto\vec \alpha}\geq 0$. 
By construction $\restr{g'}{\vec x\mapsto\vec \alpha}=\restr{g}{\vec x\mapsto\vec \alpha\oplus \vec \gamma,\ \vec u\mapsto\vec \gamma}$, and since there are no $t_i$ variables in $g$, $\restr{g}{\vec x\mapsto\vec \alpha\oplus \vec \gamma,\ \vec u\mapsto\vec \gamma}=\restr{g}{\vec x\mapsto\vec \alpha\oplus \vec \gamma,\ \vec u\mapsto\vec \gamma,\vec t\mapsto\vec \alpha}$.

Now, if $|\vec \alpha|$ is odd, in particular $\vec\alpha\neq \vec 0$ and $\vec \alpha\oplus \vec \gamma\neq \vec \gamma$.  It is easy to check that the assignment $\vec x\mapsto\vec \alpha\oplus \vec \gamma,\ \vec u\mapsto\vec \gamma,\vec t\mapsto\vec \alpha$ sets to $0$ (i.e. satisfies) all the polynomials $p\in \enc(\phi)$. Since $q$ is always non-negative on every assignment, we can conclude that $\restr{g}{\vec x\mapsto\vec \alpha\oplus \vec \gamma,\ \vec u\mapsto\vec \gamma,\vec t\mapsto\vec \alpha}\geq 0$ and therefore $\restr{g'}{\vec x\mapsto\vec \alpha}\geq 0$.
\end{proof}

\section{$\qsa$ as weighted proofs}
\label{sec:p-simulations-extra}

In propositional proof complexity it has been shown that $\SA$ and $\ns$ are equivalent to the systems \emph{weighted Resolution} \wres and \emph{restricted} weighted Resolution \reswres \cite{BBL.24}. Unsurprisingly this characterization also adapts to the QBF context. For simplicity we only do it for \wres and show the characterization for $\qsa$ as \emph{Q-weighted Resolution}.

\begin{defi}[Q-weighted Resolution]
  \label{def:QwRes}
	Let $\Phi = \mathcal{Q}. \phi$ be a false QBF. A \emph{Q-weighted Resolution} (\qwres) refutation of $\Phi$ is a sequence of configurations $\mathcal{L}_1, \dots, \mathcal{L}_k$, where each $\mathcal{L}_i$ is a multi-set of weighted clauses, \emph{i.e.} pairs of the form $(D,w)$ where $D$ is a clause and $w\in \mathbb Z$. We require that
\begin{itemize}
	\item $\mathcal{L}_1 = \emptyset,$
	\item $(\bot, w) \in \mathcal{L}_k$ for some $w > 0,$
	\item for each $(C,w)\in \mathcal{L}_k, w \ge 0,$
	\item $\mathcal{L}_{i+1}$ is obtained from $\mathcal{L}_{i}$ by the use of one of the following substitution rules:
	\begin{description}
	\item[Universal Reduction] $\dfrac{(C\lor u, 2w)}{(C, w)}$,  where all the variables in $C$ are on the left of $u$ in the quantifier prefix $\mathcal Q$ and $u$ is a universal variable,
\item[Axiom] $\dfrac{\phantom{(C,w)}}{(C,w)}$, where $C$ is a clause in $\phi$, $w\in\mathbb{Z}$,
\item[Symmetric Cut] $\dfrac{(C\lor x, w)\quad (C\lor \neg x, w)}{(C,w)}$ with $w\in \mathbb Z$,
\item[Idempotency] $\dfrac{(C\lor x \lor x, w)}{(C \lor x,w)}$ with $w\in \mathbb Z$.
	\end{description}
\end{itemize}
	The $\size$ of a \qwres refutation is the number of lines in the proof, and the Q-size is the number of universal reduction steps.
\end{defi}

We say that two configurations are equivalent, if the total weight of each clause is the same in both. For example, $\{(C, w), (C, v)\}$ is equivalent to $\{(C,w+v)\}$.
Additionally, every configuration is implicitly assumed to contain every clause with weight 0. 
These two properties combined allows one to use every clause $(C,w)$ for some rule, as long as $(C,-w)$ is also added to the next configuration.
(In \Cref{def:QwRes}, we implicitly allow that
$\mathcal{L}_{i+1}$ is equivalent to some configuration obtained from another configuration  equivalent to $\mathcal{L}_{i}$ by the use of one of the substitution rules.)

Note that every rule of \qwres can also be used in the reverse direction. The ``reverse'' rules are not required to be defined as such, because this is equivalent to using the ``forward'' rules with negative weights. 
For example, the reverse of the \emph{symmetric cut} rule is removing $(C,w)$ from the configuration and adding $(C\lor x, w)$ and $(C\lor \neg x,w)$. This is the same as adding $(C,-w)$ and removing $(C\lor x, -w)$ and $(C\lor \neg x, -w)$ due to the equivalence definition.

Recall that clause $C=\bigvee_{v\in P}v\lor \bigvee_{v\in N}\lnot v$ is encoded as the set of polynomials
\[
\enc(C) = \left\{\prod_{v\in P}\overline v\prod_{v\in N}v \right\}\cup \{v^2-v,\ v+\overline v -1 : v\in P\cup N\}\ .
\]
Let $M(C)=\prod_{v\in P}\overline v\prod_{v\in N}v$ be the monomial associated to $C$ in $\enc(C)$.

It is easy to see that multi-sets of weighted clauses and polynomials are just two ways of describing the same objects.
 The configuration $\{(C_1, a_1), ..., (C_n, a_n)\}$ corresponds to the polynomial $\sum_{i=1}^{n} -a_i \mathsf{M}(C_i)$ and the polynomial 
 \[
 \sum_{i=1}^n a_i\prod_{j\in N_i}\overline{x_j}\prod_{j\in P_i}x_j
 \]
 corresponds to the configuration $\{(\bigvee_{j\in N_i} x_j \lor \bigvee_{j\in P_i} \neg x_j, -a_i) : i\in\{1, ..., n\}\}$. This isomorphism preserves equivalence \cite{BBL.24}.

\begin{lem}[\cite{BBL.24}]
	Let configurations of weighted clauses $\mathcal{L}_1$ and $\mathcal{L}_2$ be given. $\mathcal{L}_1$ and $\mathcal{L}_2$ are equivalent if and only if 
	\begin{equation*}
		\sum_{(C,w)\in \mathcal{L}_1} w \mathsf{M}(C) = \sum_{(C,w)\in \mathcal{L}_2} w \mathsf{M}(C).
	\end{equation*}
\end{lem}

\begin{prop}
\label{prop:qsa-equiv}
	\qwres is polynomially equivalent to \qsa in both the $\size$ and $\qsize$ measure.
\end{prop}
\begin{proof}
	This proof is completely analogous to the equivalence in the propositional case between $\SA$ and $\mathsf{w}$-$\res$ \cite{BBL.24}. It suffices to show that the method can be extended for the respective reduction rules.
	
	\paragraph*{\qsa p-simulates \qwres} 
	Let $\mathcal{L}_1, ..., \mathcal{L}_t$ be a \qwres refutation of some false QBF $\mathcal{Q}.\bigwedge_{j=1}^m C_j$. 
	We construct algebraic expressions $S_1, ..., S_t$  having the form of a \qsa derivation, with the property that $S_i = -\sum_{(C,w)\in \mathcal{L}_i} w\cdot \mathsf{M}(C)$. That is, each $S_i$ is constructed as a sum of (i)~multiples of the polynomials in $\enc(C_j)$ for some $j$ (recall, this encoding includes the Boolean and the twin variable axioms), and (ii)~multiples of $(1-2u)$ for some universal variable $u$, where the multiplier polynomials depend only on variables quantified left of $u$.  
	
	Let $S_1=0$. Since $\mathcal{L}_1 = \emptyset$ this fulfils our requirement.

        Now suppose we have constructed an algebraic expression $S_i$ that equals $-\sum_{(C,w)\in \mathcal{L}_i}w\cdot\mathsf{M}(C)$. We want to construct $S_{i+1}$.  
	We differentiate between the rules used to obtain $\mathcal{L}_{i+1}$. Note that only the first case is new to the QBF variant, all the other cases are completely analogous to the propositional case. We write them anyway for completeness.
		
	\begin{itemize}
		\item If $\mathcal{L}_{i+1}$ is obtained from $\mathcal{L}_i$ by a universal reduction rule $\frac{(C\lor u, 2w)}{(C,w)}$, then add to $S_i$ the terms
		\begin{equation*}
			2w \mathsf{M}(C\lor u) - w \mathsf{M}(C) = -w \mathsf{M}(C)(1-2u).
		\end{equation*} 
		That is $S_{i+1} = S_i - w \mathsf{M}(C)(1-2u)$.
		\item If $\mathcal{L}_{i+1}$ is obtained from $\mathcal{L}_i$ by an axiom rule $\frac{}{(C,w)}$, then add to $S_i$ the term $-w \mathsf{M}(C) $.
		That is $S_{i+1} = S_i - w \mathsf{M}(C)$.
		\item If $\mathcal{L}_{i+1}$ is obtained from $\mathcal{L}_i$ by an idempotency rule $\frac{(C\lor x\lor x)}{(C\lor x,w)}$, then add to $S_i$ the terms
		\begin{equation*}
			w \mathsf{M}(C\lor x \lor x) - w \mathsf{M}(C\lor x) = w \mathsf{M}(C)(x^2-x).
		\end{equation*}
		That is $S_{i+1} = S_i + w \mathsf{M}(C)(x^2-x)$.
		\item If $\mathcal{L}_{i+1}$ is obtained from $\mathcal{L}_i$ by a symmetric cut rule $\frac{(C\lor x, w)\quad (C\lor\neg x, w)}{(C, w)}$, then add to $S_i$ the terms
		\begin{equation*}
			-w \mathsf{M}(C)+w \mathsf{M}(C\lor x) + w \mathsf{M}(C\lor\neg x) = w \mathsf{M}(C)(x + \overline{x} - 1).
		\end{equation*}
		That is $S_{i+1} = S_i + w \mathsf{M}(C)(x + \overline{x} - 1)$.
	\end{itemize}
	
	One can easily verify, that $S_{i+1}$ constructed as above, fulfils
	\begin{equation*}
		S_{i+1} = -\sum_{(C,w)\in \mathcal{L}_i}w \mathsf{M}(C)
	\end{equation*}
	and in particular $S_t = -\sum_{(C,w)\in \mathcal{L}_t} w \mathsf{M}(C)$. 
	Since $\mathcal{L}_t$ contains only an empty clause $(\bot, c)$ for some $c > 0$ and clauses with positive weights, $S_t = -c-q$ for some polynomial $q$ having only positive coefficients.
	Therefore $\frac{S_t}{c}+\frac{q}{c} + 1 = 0$ is a \qsa refutation, and its size is no  more than the size of the \qwres\ refutation.

The only times a monomial added to $S_i$ contributes to $\qsize$ is in the first case, if $\mathcal{L}_{i+1}$ is obtained by the universal reduction rule. 
	As a result, the $\qsize$ of the resulting \qsa proof is exactly the number of reduction steps, i.e.\ the $\qsize$ of the \qwres proof.
	
	\paragraph*{\qwres p-simulates \qsa}
	Let
	\begin{equation}
		\label{eq:wResSA}
		\sum_{p\in \enc(\phi)} q_pp + \sum_{u\in \mathrm{vars}_\forall(\mathcal Q)}q_u(1-2u) + q +1 = 0
	\end{equation}
	be a \qsa refutation of a false QBF $\mathcal{Q}.\phi$. We split $\enc(\phi)$ into the encoding of the clauses of $\phi$~ $P,$ the Boolean axioms $B$ and the twin variable axioms $T$, such that \eqref{eq:wResSA} is equivalent to
	\begin{equation}
		\label{eq:wResSAext}
		\sum_{p\in P} q_pp + \sum_{b\in B} q_bb + \sum_{t\in T} q_tt + \sum_{u\in \mathrm{vars}_\forall(\mathcal Q)}q_u(1-2u) + q +1 = 0.
	\end{equation}
	Without loss of generality, $q_p$ can be non-positive scalars for all $p \in P$ (cf. Lemma 5.4 in \cite{BBL.24}; the universal terms are untouched).
	Therefore, we can rewrite \eqref{eq:wResSAext} as 
	\begin{equation}
		\label{eq:wResSASplit}
		-1-q=\sum_{p\in P} -a_pp + \sum_{b\in B} q_bb + \sum_{t\in T} q_tt + \sum_{u\in \mathrm{vars}_\forall(\mathcal Q)}q_u(1-2u)
	\end{equation}
	for scalars $a_p\ge 0$.
	
	We construct a \qwres proof $(\mathcal{L}_1, ..., \mathcal{L}_t)$ of $\mathcal{Q}.\phi$.
	This is achieved by going through the right-hand side of \eqref{eq:wResSASplit} term by term and constructing sets of weighted clauses corresponding to these partial sums.
	Additionally, each set of weighted clauses can be constructed by the previous set and a \qwres rule.
	
	We start with the empty sum and set $\mathcal{L}_1 = \emptyset$. There is nothing to show.
        
	Every clause $(C_i, a_i)$ corresponds to the term $-a_i \mathsf{M}(C_i)$.
	Therefore, we can construct a sequence of configurations $\mathcal{L}_2, ..., \mathcal{L}_{k_1}$, with each 
	\[
	\mathcal{L}_i = \{(C_1,a_1), ..., (C_{i-1}, a_{i-1})\}
	\]  
	corresponding to the $i-1$-th partial sum of $\sum_{p\in P} -a_pp$. 
	In this sequence, each configuration is obtained with the previous one and the axiom rule.
	
	Continuing, for every binomial $am(x^2-x)$ from the sum $\sum_{b\in B} q_bb$ construct
	\begin{equation*}
		\mathcal{L}_{k_2} = \mathcal{L}_{k_2-1} \cup \{(C\lor\neg x\lor \neg x, -a), (C\lor\neg x, a)\}
	\end{equation*}
	where $C$ is the unique clause such that $ \mathsf{M}(C)=m$. This is an application of the idempotency rule.
	Analogously, every term $am(x_i + \overline{x}_i-1)$ from the sum $\sum_{t\in T}q_tt$ constructs
	\begin{equation*}
		\mathcal{L}_{k_3} = \mathcal{L}_{k_3-1} \cup \{(C\lor\neg x, -a),(C\lor x, -a), (C, a)\},
	\end{equation*}
	with $C$ again being the unique clause such that $ \mathsf{M}(C)=m$. This is an application of the symmetric cut rule.
	This construction is extended to the terms $am(1-2u)$ from $\sum_{u \in \vars_{\forall}(\mathcal Q)}q_u(1-2u)$.
	Here
	\begin{equation*}
	 	\mathcal{L}_{k_4} = \mathcal{L}_{k_4-1} \cup \{(C\lor\neg u, 2a), (C, -a)\}.
	 \end{equation*}
	 This corresponds to the universal reduction rule. Due to these terms being obtained from a \qsa refutation, the syntactic restrictions for the reduction rule, i.e.\ $C$ only containing variables left of $u$ and $u$ being universal, are met.
	 
	 Equation~\eqref{eq:wResSASplit} shows that this resulting configuration is equivalent to $\mathcal{L}_t = \{(\bot, 1)\} \cup \{(C_i, a_i) : -q=\sum_{i}-a_i \mathsf{M}(C_i)\}$ for clauses $C_i$ and non-negative scalars $a_i$. 
	 This representation of $q$ is possible due to the definition of \qsa.
	 As such, $(\mathcal{L}_1, ..., \mathcal{L}_t)$ is a \qwres refutation of $\mathcal{Q}.\phi$ with rational weights, and a suitable multiple gives an integral version matching \Cref{def:QwRes}.
	 For every monomial in $q_u$, exactly one universal reduction is added. As such, the Q-sizes match.
\end{proof}

\section{Simulations}
\label{sec:p-simulations}

We now investigate how the algebraic QBF systems relate to each other and to other known QBF proof systems such as QU-Resolution and Q-PC and  show the p-simulations of \Cref{fig:p-sim}.

\begin{thm}
	\label{thm:psim-sos-sa}
	\qsos p-simulates \qsa w.r.t.\ the $\size$ and $\qsize$ measures.
\end{thm}

\begin{proof}
It is well known that (degree $2d$) $\sos$ p-simulates (degree $d$) $\SA$ (see for instance \cite[Lemma 3.63]{FKP.19}). The same argument works without change in the QBF setting: every variable~$v$ in a positive monomial can be substituted by $v^2$ summing a suitable multiple of $v^2-v$. In this way, every positive monomial $\frac{a}{b}m$ with $a,b\in \mathbb N$ can be converted into $\frac{s_1^2+s_2^2+s_3^2+s_4^2}{b^2}m^2$ where $s_1,s_2,s_3,s_4$ are four integers that sum up to $ab$ (they exists by Lagrange's Four Squares Theorem). This converts $\frac{a}{b}m$ into a sum of at most four squares with rational coefficients. 
\end{proof}

In the argument above, notice that the only increment in degree is in the propositional part, hence, different from  the propositional case, \qsos\ with $\qdeg$ $d$ p-simulates \qsa\ with $\qdeg$ $d$. In the $\qsize$ measure, the converse also holds.

\begin{thm}
	\label{thm:pequiv-semialg}
	$\qsa$ p-simulates $\qsos$ w.r.t.\ the $\qsize$ measure.
\end{thm}
\begin{proof}
Every polynomial $q$ that is non-negative on the Boolean assignments can be written as a (possibly exponentially large) sum of the form $\sum_{\alpha}\restr{q}{\alpha}\chi_{\alpha}(\vec v)$, where $\chi_\alpha(\vec v)$ is a monomial that evaluates to $1$ when the variables are set according to the Boolean assignment $\alpha$ and on any other Boolean assignment it is $0$. In other words, every $\qsos$ refutation can be written as a possibly exponentially larger $\qsa$ refutation. The exponential blow-up appears in the propositional part which is not accounted for in $\qsize$.
\end{proof}

\begin{thm}
	\label{thm:psim-sa-res}
	\qsa p-simulates  \qures  w.r.t.\ both $\size$ and $\qsize$.
\end{thm}
\begin{proof}
Due to \Cref{prop:qsa-equiv}, to show that $\qsa$ p-simulates $\qures$ it is enough to show that \qwres p-simulates $\qures$. This p-simulation is an adaptation of the proof that $\mathsf{w}$-$\res$ p-simulates $\res$ \cite{BBL.24}.
	This proof is done in two steps. First, we show that every resolution step in a \qures refutation can be transformed into a sequence of \emph{weakening} and \emph{symmetric cut} rules. 
	Second, we then add appropriate weights to the clauses to transform it into a \qwres refutation. 
	
	Let
	\begin{equation*}
		\frac{(C\lor D \lor x)\quad (C\lor E \lor \neg x)}{(C \lor D \lor E)}
	\end{equation*}
	be an arbitrary application of the resolution rule in a \qures refutation, where $C,D,E$ are clauses with pairwise disjoint variables and none containing $x$ or $\neg x$. Let $D$ consist of literals $d_1, d_2,... ,d_k$ and $E$ of literals $e_1, e_2, ..., e_\ell$.
	
	Using $\ell$ consecutive \emph{weakening} rules, we expand $(C\lor D \lor x)$ to $(C \lor D\lor E \lor x)$ and $(C\lor E \lor \neg x)$ to $(C\lor D \lor E \lor \neg x)$. 
	These two clauses can then be used in a single \emph{symmetric cut} to gain $(C\lor D \lor E)$.
	
	For the second step, we need to add weights to the clauses. Since \qures proofs in general are not treelike, we need to chose the weights in a way that keeps all previous clauses in the configuration.
	
	Let $\pi= C_1, ..., C_t$ be an arbitrary \qures refutation of a false QBF, using symmetric cut instead of resolution. 
	We inductively construct a sequence of configurations $\mathcal{L}_1, ..., \mathcal{L}_t$, such that these are a valid \qwres proof and for every $i \leq t$ there exist strictly positive weights $w_1^i, ..., w_i^i$ with $\{(C_j, w_j^i),  j \in \{1,...,i\}\} \subseteq \mathcal{L}_i$. Additionally, all the weights in every $\mathcal{L}_i$ are non-negative.        
	
	We differentiate between the rules used to obtain $C_i$.
	\begin{itemize}
		\item $C_i$ is an axiom. Then $\mathcal{L}_{i} = \mathcal{L}_{i-1} \cup \{(C_i,1)\}$ and $\mathcal{L}_i$ is obtained from $\mathcal{L}_{i-1}$ using the axiom rule.
		\item $C_i$ is obtained from $C_j$ using the universal reduction rule, i.e. $C_j = C_i \lor u$ for some universal literal $u$. Let $w_j$ be the weight such that $(C_j, w_j) \in \mathcal{L}_{i-1}$. Then \[\mathcal{L}_{i} = \mathcal{L}_{i-1} \cup \{(C_j, -\frac{w_j}{2})\} \cup \{(C_i,\frac{w_j}{4})\}\] and $\mathcal{L}_i$ is obtained from $\mathcal{L}_{i-1}$ using the universal reduction rule.
		\item $C_i$ is obtained from $C_j$ using the weakening rule, i.e. $C_i = C_j \lor v$ for some literal $v$. Let $w_j$ be the weight such that $(C_j, w_j) \in \mathcal{L}_{i-1}$. Then \[\mathcal{L}_{i} = \mathcal{L}_{i-1} \cup \{(C_j, -\frac{w_j}{2})\} \cup \{(C_i,\frac{w_j}{2}), (C_j \lor \neg v,\frac{w_j}{2})\}\] and $\mathcal{L}_i$ is obtained from $\mathcal{L}_{i-1}$ using the symmetric rule.
		\item $C_i$ is obtained from $C_j$ and $C_k$ using the symmetric cut rule, i.e. $C_j = C_i \lor v$ and $C_k = C_i \lor \neg v$ for some variable $v$. Let $w_j$ and $w_k$ be the weights such that $(C_j, w_j), (C_k,w_k) \in \mathcal{L}_{i-1}$ and define $w = \min\{w_j, w_k\}$. Then \[\mathcal{L}_{i} = \mathcal{L}_{i-1} \cup \{(C_j, -\frac{w}{2}), (C_k, -\frac{w}{2})\} \cup \{(C_i,\frac{w}{2})\}\] and $\mathcal{L}_i$ is obtained from $\mathcal{L}_{i-1}$ using the symmetric cut rule.
	\end{itemize}
	The non-negativity of all weights follow directly. 
	Since $\pi$ is a \qures proof, $C_t = \bot$. As such, there exists some weight $w_t > 0$ with $(\bot, w_t) \in \mathcal{L}_t$. 
        Thus $\mathcal{L}_1, ..., \mathcal{L}_t$ is a valid \qwres refutation with rational weights, and a suitable multiple gives an integral version matching \Cref{def:QwRes}.
	
	In total, for every resolution rule of $\pi$ we add at most $n$ lines to the $\qwres$ proof and for every reduction rule we add exactly one line (where $n$ is the number of variables). The resulting proof has size at most $n|\pi|$, and its Q-size is $\qsize(\pi)$.
\end{proof}

\begin{thm}
\label{thm:psim-pc-ns}
	$\qpc$ p-simulates $\qns$ w.r.t.\ the $\size$ and $\qsize$ measures.
\end{thm}
\begin{proof}
	Let $\mathcal Q.\phi$ be a false QBF with $n$ variables and $m$ clauses and $\pi$ be a \qns refutation of $\mathcal Q.\phi$ of the form
	\begin{equation}
\label{eq:NS-psim}
	\sum_{p\in \enc(\phi)} q_p p+ \sum_{u\in \mathrm{vars}_\forall(\mathcal Q)}q_u(1-2u) + 1=0\ ,
\end{equation}
Lines in a \qpc proof can be multiplied by arbitrary polynomials. Hence, we can obtain in $\qpc$ the sum $\sum_{p\in \enc(\phi)}q_pp$ from the polynomials in $\enc(\phi)$ in a polynomial number of steps. Let $\vars_\forall(\mathcal Q)$ be $u_1,u_2,\dots,u_n$.
	Due to the symbolic equality in eq.~\eqref{eq:NS-psim}, this sum equals $-1 - \sum_{i=1}^{n}q_{u_i}(1-2u_i)$. We then use the  $\forall$-reduction  on $u_n$, then $u_{n-1}$ etc. In the first step, restricting by $u_n=1$ and $u_n=0$, we get respectively
	\[
	-1 - \sum_{i=1}^{n-1}q_{u_i}(1-2u_i) - q_{u_{n}} \text{\ and\ } -1 - \sum_{i=1}^{n-1}q_{u_i}(1-2u_i) + q_{u_{n}}\ .
	\]
	Adding them, we get $-1 - \sum_{i=1}^{n-1}q_{u_i}(1-2u_i)$. We repeat this process until we get rid of all universal variables and only the $-1$ remains. It is clear from the argument that this simulation only increases the $\size$ and $\qsize$ linearly.
\end{proof}

\begin{thm}
	\label{thm:psim-sos-pc}
	\qsos p-simulates  \qpc w.r.t.\ the $\size$ and $\qsize$ measures.
\end{thm}

\begin{proof}(sketch)
The argument in \cite[Lemma 3.1]{Berkholz18} showing that degree $2d$ $\sos$ p-simulates degree $d$ $\pc$ adapts easily to the QBF setting. The idea is that given a $\qpc$ derivation $p_1,\dots, p_s$ we derive an algebraic expression for $-p_i^2$ which eventually for $i=s$ will give a \qsos refutation of $\mathcal Q.\phi$. This is done inductively on $i$ and is based on the following algebraic identities:
\begin{itemize}
\item sum rule (from $p$ and $q$ deduce $ap+bq$ with $a,b\in \mathbb Q$):
  \[-(ap+bq)^2=-2a^2p^2-2b^2q^2+(ap-bq)^2\ ;
  \]
\item product rule (from $p$ deduce $xp$):
  \[-(xp)^2= -p^2+(p-xp)^2+2p^2(x-x^2)\ ;
  \]
	\item $\forall$-reduction (from $p+qu$ deduce $p$):
          \[
          -p^2=-2(p+qu)^2+(p+q)^2-(q^2+2pq)(1-2u)+2q^2(u^2-u) \ ;
          \]
	\item $\forall$-reduction (from $p+qu$ deduce $p+q$):
          \[
          -(p+q)^2=-2(p+qu)^2+p^2-(q^2+2pq)(1-2u)+2q^2(u^2-u) \ .
          \] \qedhere
\end{itemize}
\end{proof}

\begin{cor}
	\qsos p-simulates \qns and is exponentially stronger, w.r.t.\ the $\size$ and $\qsize$ measures.
\end{cor}

\begin{proof}
	The simulation follows from \cref{thm:psim-sos-pc} and \cref{thm:psim-pc-ns}. The separation follows from \cref{prop:majority} and \cref{thm:psim-pc-ns} for the $\qsize$ measure, and from propositional separations \cite{ALN.16,DGGM24} for the $\size$ measure.
\end{proof}

\section{Conclusion}
\label{sec:conclusions}

In this work we defined semi-algebraic proof systems for QBF and initiated their proof complexity investigation. While our results already reveal an interesting picture in terms of simulations and lower and upper bounds, a number of questions remain that appear to be of interest for further research.

 In the propositional setting $\res$ and $\ns$ are incomparable proof systems. Are also $\qns$ and $\qures$ incomparable w.r.t.\ the $\qsize$ measure?
	
	 In \cref{sec:strategy} we showed how to express strategy extraction for $\qsos$ using polynomial threshold functions. Although this suffices for lower bounds, it appears interesting to determine the correct computational model \emph{characterizing}  strategy extraction for \qsos and \qns in the same sense as the tight characterisations for $\qures$ \cite{BBMP22}, $\qpc$ \cite{BeyersdorffHKS24}, and QBF Frege systems \cite{BBCP20}. 
	 
	 Finally, it would be interesting to determine the relationship of our new semi-algebraic QBF systems to the static expansion-based  algebraic systems suggested in \cite{CCKS23}, which might  turn out to be incomparable in strength.
\sbox0{ \cite{ALN.16-conf,BBC16,BP16,BL20,DGM.20} }

\section*{Acknowledgment}
  \noindent The authors  thank the Oberwolfach Research Institute for Mathematics: the idea for this work started during the Oberwolfach workshop 2413 \emph{Proof Complexity and Beyond}, and the Schloss Dagstuhl – Leibniz Center for Informatics: part of this work has been done during the Dagstuhl Seminar 24421 \emph{SAT and Interactions}.

\newcommand{\etalchar}[1]{$^{#1}$}


\begin{thebibliography}{DGGM24}

\bibitem[AFT11]{AtseriasFT11}
Albert Atserias, Johannes~Klaus Fichte, and Marc Thurley.
\newblock Clause-learning algorithms with many restarts and bounded-width
  resolution.
\newblock {\em J. Artif. Intell. Res.}, 40:353--373, 2011.
\newblock \href {https://doi.org/10.1613/jair.3152}
  {\path{doi:10.1613/jair.3152}}.

\bibitem[AH19]{AH.19}
Albert Atserias and Tuomas Hakoniemi.
\newblock {Size-Degree Trade-Offs for Sums-of-Squares and Positivstellensatz
  Proofs}.
\newblock In {\em 34th Computational Complexity Conference (CCC 2019)}, volume
  137, pages 24:1--24:20, 2019.
\newblock \href {https://doi.org/10.4230/LIPICS.CCC.2019.24}
  {\path{doi:10.4230/LIPICS.CCC.2019.24}}.

\bibitem[ALN14]{ALN.16-conf}
Albert Atserias, Massimo Lauria, and Jakob Nordstr{\"{o}}m.
\newblock Narrow proofs may be maximally long.
\newblock In {\em 29th Conference on Computational Complexity (CCC 2014)},
  pages 286--297, 2014.
\newblock \href {https://doi.org/10.1109/CCC.2014.36}
  {\path{doi:10.1109/CCC.2014.36}}.

\bibitem[ALN16]{ALN.16}
Albert Atserias, Massimo Lauria, and Jakob Nordstr{\"{o}}m.
\newblock Narrow proofs may be maximally long.
\newblock {\em ACM Trans. Comput. Logic}, 17(3):19:1--19:30, 2016.
\newblock A preliminary version of this work appeared as \cite{ALN.16-conf}.
\newblock \href {https://doi.org/10.1145/2898435} {\path{doi:10.1145/2898435}}.

\bibitem[BB23]{BB23-LMCS}
Olaf Beyersdorff and Benjamin B{\"{o}}hm.
\newblock Understanding the relative strength of {QBF} {CDCL} solvers and {QBF}
  resolution.
\newblock {\em Log. Methods Comput. Sci.}, 19(2), 2023.
\newblock \href {https://doi.org/10.46298/lmcs-19(2:2)2023}
  {\path{doi:10.46298/lmcs-19(2:2)2023}}.

\bibitem[BBC16]{BBC16}
Olaf Beyersdorff, Ilario Bonacina, and Leroy Chew.
\newblock Lower bounds: From circuits to {QBF} proof systems.
\newblock In {\em Proc.\ ACM Conference on Innovations in Theoretical Computer
  Science (ITCS'16)}, pages 249--260, 2016.
\newblock \href {https://doi.org/10.1145/2840728.2840740}
  {\path{doi:10.1145/2840728.2840740}}.

\bibitem[BBCP20]{BBCP20}
Olaf Beyersdorff, Ilario Bonacina, Leroy Chew, and Jan Pich.
\newblock Frege systems for quantified {Boolean} logic.
\newblock {\em J. ACM}, 67(2):9:1--9:36, 2020.
\newblock Preliminary versions of this work appeared as \cite{BBC16} and
  \cite{BP16}.
\newblock \href {https://doi.org/10.1145/3381881} {\path{doi:10.1145/3381881}}.

\bibitem[BBH{\etalchar{+}}12]{BBHKSZ12}
Boaz Barak, Fernando~G.S.L. Brandao, Aram~W. Harrow, Jonathan Kelner, David
  Steurer, and Yuan Zhou.
\newblock Hypercontractivity, sum-of-squares proofs, and their applications.
\newblock In {\em 44th Annual ACM Symposium on Theory of Computing (STOC
  2012)}, pages 307--326, 2012.
\newblock \href {https://doi.org/10.1145/2213977.2214006}
  {\path{doi:10.1145/2213977.2214006}}.

\bibitem[BBH18]{BBH18}
Olaf Beyersdorff, Joshua Blinkhorn, and Luke Hinde.
\newblock Size, cost, and capacity: A semantic technique for hard random
  {QBFs}.
\newblock In {\em Proc.\ Conference on Innovations in Theoretical Computer
  Science (ITCS'18)}, pages 9:1--9:18, 2018.
\newblock \href {https://doi.org/10.4230/LIPICS.ITCS.2018.9}
  {\path{doi:10.4230/LIPICS.ITCS.2018.9}}.

\bibitem[BBH19]{BBH19}
Olaf Beyersdorff, Joshua Blinkhorn, and Luke Hinde.
\newblock Size, cost, and capacity: {A} semantic technique for hard random
  {QBFs}.
\newblock {\em Logical Methods in Computer Science}, 15(1), 2019.
\newblock A preliminary version of this work appeared as \cite{BBH18}.
\newblock \href {https://doi.org/10.23638/LMCS-15(1:13)2019}
  {\path{doi:10.23638/LMCS-15(1:13)2019}}.

\bibitem[BBK{\etalchar{+}}25]{BBKMS-SAT-25}
Olaf Beyersdorff, Ilario Bonacina, Kaspar Kasche, Meena Mahajan, and
  Luc~Nicolas Spachmann.
\newblock {Semi-Algebraic Proof Systems for QBF}.
\newblock In {\em 28th International Conference on Theory and Applications of
  Satisfiability Testing (SAT 2025)}, volume 341, pages 5:1--5:19. LIPIcs,
  2025.
\newblock \href {https://doi.org/10.4230/LIPIcs.SAT.2025.5}
  {\path{doi:10.4230/LIPIcs.SAT.2025.5}}.

\bibitem[BBL24]{BBL.24}
Ilario Bonacina, Maria~Luisa Bonet, and Jordi Levy.
\newblock Weighted, circular and semi-algebraic proofs.
\newblock {\em J. Artif. Intell. Res.}, 79:447--482, 2024.
\newblock A preliminary version of this work appeared as \cite{BL20}.
\newblock \href {https://doi.org/10.1613/JAIR.1.15075}
  {\path{doi:10.1613/JAIR.1.15075}}.

\bibitem[BBMP23]{BBMP22}
Olaf Beyersdorff, Joshua Blinkhorn, Meena Mahajan, and Tom{\'{a}}s Peitl.
\newblock Hardness characterisations and size-width lower bounds for {QBF}
  resolution.
\newblock {\em {ACM} Trans. Comput. Log.}, 24(2):10:1--10:30, 2023.
\newblock \href {https://doi.org/10.1145/3565286} {\path{doi:10.1145/3565286}}.

\bibitem[BCJ15]{BCJ15}
Olaf Beyersdorff, Leroy Chew, and Mikol{\'a}\v{s} Janota.
\newblock Proof complexity of resolution-based {QBF} calculi.
\newblock In {\em Proc.\ Symposium on Theoretical Aspects of Computer Science
  (STACS'15)}, pages 76--89. LIPIcs, 2015.
\newblock \href {https://doi.org/10.4230/LIPIcs.STACS.2015.76}
  {\path{doi:10.4230/LIPIcs.STACS.2015.76}}.

\bibitem[BCJ19]{BeyersdorffCJ19}
Olaf Beyersdorff, Leroy Chew, and Mikol{\'{a}}s Janota.
\newblock New resolution-based {QBF} calculi and their proof complexity.
\newblock {\em ACM Transactions on Computation Theory}, 11(4):26:1--26:42,
  2019.
\newblock \href {https://doi.org/10.1145/3352155} {\path{doi:10.1145/3352155}}.

\bibitem[BCMS18]{BCMS18-CP}
Olaf Beyersdorff, Leroy Chew, Meena Mahajan, and Anil Shukla.
\newblock Understanding cutting planes for {QBFs}.
\newblock {\em Inf. Comput.}, 262:141--161, 2018.
\newblock \href {https://doi.org/10.1016/j.ic.2018.08.002}
  {\path{doi:10.1016/j.ic.2018.08.002}}.

\bibitem[Ber18]{Berkholz18}
Christoph Berkholz.
\newblock {The Relation between Polynomial Calculus, Sherali-Adams, and
  Sum-of-Squares Proofs}.
\newblock In {\em 35th Symposium on Theoretical Aspects of Computer Science
  (STACS)}, volume~96, pages 11:1--11:14, 2018.
\newblock \href {https://doi.org/10.4230/LIPIcs.STACS.2018.11}
  {\path{doi:10.4230/LIPIcs.STACS.2018.11}}.

\bibitem[Bey22]{M4CQBF}
Olaf Beyersdorff.
\newblock Proof complexity of quantified {Boolean} logic -- a survey.
\newblock In Marco Benini, Olaf Beyersdorff, Michael Rathjen, and Peter
  Schuster, editors, {\em Mathematics for Computation (M4C)}, pages 353--391.
  World Scientific, Singapore, 2022.

\bibitem[BGIP01]{BGIP.01}
Sam Buss, Dima Grigoriev, Russell Impagliazzo, and Toniann Pitassi.
\newblock Linear gaps between degrees for the polynomial calculus modulo
  distinct primes.
\newblock {\em Journal of Computer and System Sciences}, 62(2):267 -- 289,
  2001.
\newblock \href {https://doi.org/10.1006/jcss.2000.1726}
  {\path{doi:10.1006/jcss.2000.1726}}.

\bibitem[BHKS24]{BeyersdorffHKS24}
Olaf Beyersdorff, Tim Hoffmann, Kaspar Kasche, and Luc~Nicolas Spachmann.
\newblock Polynomial calculus for quantified boolean logic: Lower bounds
  through circuits and degree.
\newblock In {\em 49th International Symposium on Mathematical Foundations of
  Computer Science (MFCS)}, volume 306, pages 27:1--27:15, 2024.
\newblock \href {https://doi.org/10.4230/LIPICS.MFCS.2024.27}
  {\path{doi:10.4230/LIPICS.MFCS.2024.27}}.

\bibitem[BHP20]{BHP20}
Olaf Beyersdorff, Luke Hinde, and J\'{a}n Pich.
\newblock Reasons for hardness in {QBF} proof systems.
\newblock {\em ACM Transactions on Computation Theory}, 12(2):10:1--10:27,
  2020.
\newblock \href {https://doi.org/10.1145/3378665} {\path{doi:10.1145/3378665}}.

\bibitem[BJ12]{Balabanov12}
Valeriy Balabanov and Jie-Hong~R. Jiang.
\newblock Unified {QBF} certification and its applications.
\newblock {\em Form. Methods Syst. Des.}, 41(1):45--65, 2012.
\newblock \href {https://doi.org/10.1007/s10703-012-0152-6}
  {\path{doi:10.1007/s10703-012-0152-6}}.

\bibitem[BJLS21]{qbfhandbook}
Olaf Beyersdorff, Mikol{\'{a}}s Janota, Florian Lonsing, and Martina Seidl.
\newblock Quantified {Boolean} formulas.
\newblock In Armin Biere, Marijn Heule, Hans van Maaren, and Toby Walsh,
  editors, {\em Handbook of Satisfiability, 2nd edition}, Frontiers in
  Artificial Intelligence and Applications. IOS press, 2021.
\newblock \href {https://doi.org/10.3233/FAIA201015}
  {\path{doi:10.3233/FAIA201015}}.

\bibitem[BL20]{BL20}
Maria~Luisa Bonet and Jordi Levy.
\newblock Equivalence between systems stronger than resolution.
\newblock In {\em 23rd International Conference on Theory and Applications of
  Satisfiability Testing (SAT 2020)}, volume 12178, pages 166--181, 2020.
\newblock \href {https://doi.org/10.1007/978-3-030-51825-7_13}
  {\path{doi:10.1007/978-3-030-51825-7_13}}.

\bibitem[Bla37]{Bla37}
A.~Blake.
\newblock {\em Canonical expressions in boolean algebra}.
\newblock PhD thesis, University of Chicago, 1937.

\bibitem[BN21]{sathandbookpc}
Sam Buss and Jakob Nordstr{\"{o}}m.
\newblock Proof complexity and {SAT} solving.
\newblock In Armin Biere, Marijn Heule, Hans van Maaren, and Toby Walsh,
  editors, {\em Handbook of Satisfiability}, Frontiers in Artificial
  Intelligence and Applications, pages 233--350. {IOS} Press, 2021.
\newblock \href {https://doi.org/10.3233/FAIA200990}
  {\path{doi:10.3233/FAIA200990}}.

\bibitem[BP16]{BP16}
Olaf Beyersdorff and J{\'{a}}n Pich.
\newblock Understanding {Gentzen} and {Frege} systems for {QBF}.
\newblock In {\em Proc.\ ACM/IEEE Symposium on Logic in Computer Science
  (LICS)}, 2016.
\newblock \href {https://doi.org/10.1145/2933575.2933597}
  {\path{doi:10.1145/2933575.2933597}}.

\bibitem[BSW01]{BW01}
Eli Ben-Sasson and Avi Wigderson.
\newblock Short proofs are narrow - resolution made simple.
\newblock {\em J. ACM}, 48(2):149--169, 2001.
\newblock \href {https://doi.org/10.1145/375827.375835}
  {\path{doi:10.1145/375827.375835}}.

\bibitem[BWJ14]{BWJ14}
Valeriy Balabanov, Magdalena Widl, and Jie-Hong~R. Jiang.
\newblock {QBF} resolution systems and their proof complexities.
\newblock In {\em Proc.\ Theory and Applications of Satisfiability Testing
  {(SAT)}}, pages 154--169, 2014.
\newblock \href {https://doi.org/10.1007/978-3-319-09284-3_12}
  {\path{doi:10.1007/978-3-319-09284-3_12}}.

\bibitem[CCKS23]{CCKS23}
Sravanthi Chede, Leroy Chew, Balesh Kumar, and Anil Shukla.
\newblock Understanding {Nullstellensatz} for {QBF}s.
\newblock {\em Electron. Colloquium Comput. Complex.}, {TR23-129}, 2023.
\newblock URL: \url{https:TR//eccc.weizmann.ac.il/report/2023/129/}.

\bibitem[CCT87]{CCT87}
William Cook, Collette~R. Coullard, and Gy{\"{o}}rgy Tur{\'{a}}n.
\newblock On the complexity of cutting-plane proofs.
\newblock {\em Discrete Applied Mathematics}, 18(1):25--38, 1987.
\newblock \href {https://doi.org/10.1016/0166-218X(87)90039-4}
  {\path{doi:10.1016/0166-218X(87)90039-4}}.

\bibitem[CEI96]{CEI.96}
Matthew Clegg, Jeff Edmonds, and Russell Impagliazzo.
\newblock Using the {Groebner} basis algorithm to find proofs of
  unsatisfiability.
\newblock In {\em 28th Annual {ACM} Symposium on the Theory of Computing
  (STOC)}, pages 174--183, 1996.
\newblock \href {https://doi.org/10.1145/237814.237860}
  {\path{doi:10.1145/237814.237860}}.

\bibitem[Che17]{Che17}
Hubie Chen.
\newblock Proof complexity modulo the polynomial hierarchy: Understanding
  alternation as a source of hardness.
\newblock {\em ACM Transactions on Computation Theory}, 9(3):15:1--15:20, 2017.
\newblock \href {https://doi.org/10.1145/3087534} {\path{doi:10.1145/3087534}}.

\bibitem[CR79]{CR79}
Stephen~A. Cook and Robert~A. Reckhow.
\newblock The relative efficiency of propositional proof systems.
\newblock {\em The Journal of Symbolic Logic}, 44(1):36--50, 1979.
\newblock \href {https://doi.org/10.2307/2273702} {\path{doi:10.2307/2273702}}.

\bibitem[DGGM24]{DGGM24}
Stefan~S. Dantchev, Nicola Galesi, Abdul Ghani, and Barnaby Martin.
\newblock Proof complexity and the binary encoding of combinatorial principles.
\newblock {\em {SIAM} J. Comput.}, 53(3):764--802, 2024.
\newblock A preliminary version appeared as \cite{DGM.20}.
\newblock \href {https://doi.org/10.1137/20M134784X}
  {\path{doi:10.1137/20M134784X}}.

\bibitem[DGM20]{DGM.20}
Stefan~S. Dantchev, Abdul Ghani, and Barnaby Martin.
\newblock Sherali-adams and the binary encoding of combinatorial principles.
\newblock In {\em 14th Latin American Symposium on Theoretical Informatics
  (LATIN 2020)}, volume 12118, pages 336--347, 2020.
\newblock \href {https://doi.org/10.1007/978-3-030-61792-9_27}
  {\path{doi:10.1007/978-3-030-61792-9_27}}.

\bibitem[DM13]{DantchevM13}
Stefan~S. Dantchev and Barnaby Martin.
\newblock Rank complexity gap for lov{\'{a}}sz-schrijver and sherali-adams
  proof systems.
\newblock {\em Comput. Complex.}, 22(1):191--213, 2013.
\newblock \href {https://doi.org/10.1007/S00037-012-0049-1}
  {\path{doi:10.1007/S00037-012-0049-1}}.

\bibitem[DMR09]{DMR.09}
Stefan~S. Dantchev, Barnaby Martin, and Mark Nicholas~Charles Rhodes.
\newblock Tight rank lower bounds for the {Sherali-Adams} proof system.
\newblock {\em Theor. Comput. Sci.}, 410(21-23):2054--2063, 2009.
\newblock \href {https://doi.org/10.1016/J.TCS.2009.01.002}
  {\path{doi:10.1016/J.TCS.2009.01.002}}.

\bibitem[dRLNS21]{RezendeLN021}
Susanna~F. de~Rezende, Massimo Lauria, Jakob Nordstr{\"{o}}m, and Dmitry
  Sokolov.
\newblock The power of negative reasoning.
\newblock In {\em 36th Computational Complexity Conference (CCC)}, volume 200,
  pages 40:1--40:24, 2021.
\newblock \href {https://doi.org/10.4230/LIPICS.CCC.2021.40}
  {\path{doi:10.4230/LIPICS.CCC.2021.40}}.

\bibitem[ELW13]{ELW13}
Uwe Egly, Florian Lonsing, and Magdalena Widl.
\newblock Long-distance resolution: Proof generation and strategy extraction in
  search-based {QBF} solving.
\newblock In {\em Proc. Logic for Programming, Artificial Intelligence, and
  Reasoning (LPAR)}, pages 291--308, 2013.
\newblock \href {https://doi.org/10.1007/978-3-642-45221-5_21}
  {\path{doi:10.1007/978-3-642-45221-5_21}}.

\bibitem[FKP19]{FKP.19}
Noah Fleming, Pravesh Kothari, and Toniann Pitassi.
\newblock Semialgebraic proofs and efficient algorithm design.
\newblock {\em Found. Trends Theor. Comput. Sci.}, 14(1-2):1--221, 2019.
\newblock \href {https://doi.org/10.1561/0400000086}
  {\path{doi:10.1561/0400000086}}.

\bibitem[GHJ{\etalchar{+}}24]{GoosHJMPRT24}
Mika G{\"{o}}{\"{o}}s, Alexandros Hollender, Siddhartha Jain, Gilbert Maystre,
  William Pires, Robert Robere, and Ran Tao.
\newblock Separations in proof complexity and {TFNP}.
\newblock {\em J. {ACM}}, 71(4):26:1--26:45, 2024.
\newblock \href {https://doi.org/10.1145/3663758} {\path{doi:10.1145/3663758}}.

\bibitem[Hak20]{Hakoniemi20}
Tuomas Hakoniemi.
\newblock Feasible interpolation for polynomial calculus and sums-of-squares.
\newblock In {\em 47th International Colloquium on Automata, Languages, and
  Programming (ICALP)}, volume 168, pages 63:1--63:14, 2020.
\newblock \href {https://doi.org/10.4230/LIPICS.ICALP.2020.63}
  {\path{doi:10.4230/LIPICS.ICALP.2020.63}}.

\bibitem[JM15]{JM15}
Mikol{\'{a}}s Janota and Joao Marques{-}Silva.
\newblock Expansion-based {QBF} solving versus {Q}-resolution.
\newblock {\em Theor. Comput. Sci.}, 577:25--42, 2015.
\newblock \href {https://doi.org/10.1016/j.tcs.2015.01.048}
  {\path{doi:10.1016/j.tcs.2015.01.048}}.

\bibitem[KKF95]{KBKF95}
Hans {Kleine B{\"u}ning}, Marek Karpinski, and Andreas Fl{\"o}gel.
\newblock Resolution for quantified {Boolean} formulas.
\newblock {\em Inf. Comput.}, 117(1):12--18, 1995.
\newblock \href {https://doi.org/10.1006/INCO.1995.1025}
  {\path{doi:10.1006/INCO.1995.1025}}.

\bibitem[Las01]{Lasserre01}
Jean~B. Lasserre.
\newblock Global optimization with polynomials and the problem of moments.
\newblock {\em {SIAM} J. Optim.}, 11(3):796--817, 2001.
\newblock \href {https://doi.org/10.1137/S1052623400366802}
  {\path{doi:10.1137/S1052623400366802}}.

\bibitem[MP87]{MP87}
Marvin Minsky and Seymour Papert.
\newblock {\em Perceptrons - an introduction to computational geometry}.
\newblock {MIT} Press, 1987.
\newblock \href {https://doi.org/10.7551/mitpress/11301.001.0001}
  {\path{doi:10.7551/mitpress/11301.001.0001}}.

\bibitem[OZ13]{ODonnellZ13}
Ryan O'Donnell and Yuan Zhou.
\newblock Approximability and proof complexity.
\newblock In {\em Proceedings of the Twenty-Fourth Annual {ACM-SIAM} Symposium
  on Discrete Algorithms (SODA)}, pages 1537--1556. {SIAM}, 2013.
\newblock \href {https://doi.org/10.1137/1.9781611973105.111}
  {\path{doi:10.1137/1.9781611973105.111}}.

\bibitem[PD11]{DBLP:journals/ai/PipatsrisawatD11}
Knot Pipatsrisawat and Adnan Darwiche.
\newblock On the power of clause-learning {SAT} solvers as resolution engines.
\newblock {\em Artif. Intell.}, 175(2):512--525, 2011.
\newblock \href {https://doi.org/10.1016/j.artint.2010.10.002}
  {\path{doi:10.1016/j.artint.2010.10.002}}.

\bibitem[Rob65]{Rob65}
John~Alan Robinson.
\newblock A machine-oriented logic based on the resolution principle.
\newblock {\em Journal of the ACM}, 12:23--41, 1965.
\newblock \href {https://doi.org/10.1145/321250.321253}
  {\path{doi:10.1145/321250.321253}}.

\bibitem[SA90]{SheraliA90}
Hanif~D. Sherali and Warren~P. Adams.
\newblock A hierarchy of relaxations between the continuous and convex hull
  representations for zero-one programming problems.
\newblock {\em {SIAM} J. Discret. Math.}, 3(3):411--430, 1990.
\newblock \href {https://doi.org/10.1137/0403036} {\path{doi:10.1137/0403036}}.

\bibitem[Sok20]{Sokolov20}
Dmitry Sokolov.
\newblock (semi)algebraic proofs over {\(\pm\)}1 variables.
\newblock In {\em Proceedings of the 52nd Annual {ACM} {SIGACT} Symposium on
  Theory of Computing (STOC)}, pages 78--90, 2020.
\newblock \href {https://doi.org/10.1145/3357713.3384288}
  {\path{doi:10.1145/3357713.3384288}}.

\bibitem[VG12]{Gelder12}
Allen Van~Gelder.
\newblock Contributions to the theory of practical quantified {Boolean} formula
  solving.
\newblock In {\em Proc.\ Principles and Practice of Constraint Programming
  (CP)}, pages 647--663, 2012.
\newblock \href {https://doi.org/10.1007/978-3-642-33558-7_47}
  {\path{doi:10.1007/978-3-642-33558-7_47}}.

\bibitem[ZM02]{ZM02}
Lintao Zhang and Sharad Malik.
\newblock Conflict driven learning in a quantified {Boolean} satisfiability
  solver.
\newblock In {\em Proc. {IEEE/ACM} International Conference on Computer-aided
  Design {(ICCAD)}}, pages 442--449, 2002.
\newblock \href {https://doi.org/10.1145/774572.774637}
  {\path{doi:10.1145/774572.774637}}.

\end{thebibliography}
\end{document}